\documentclass[twoside,11pt,final]{entics}
\usepackage{quiver}
\usepackage{macros-postproc}
\usepackage{enticsmacro}
\usepackage{graphicx}
\usepackage{caption}
\usepackage{subcaption}
\usepackage{wrapfig}
\usepackage[all]{xy}

\usepackage{tikz, amsmath}
\usetikzlibrary{cd}
\makeatletter
\tikzcdset{
  eq node/.style={
    commutative diagrams/math mode=false, anchor=center},
  eq/.style={
    phantom,
    /tikz/every to/.append style={
      edge node={node[commutative diagrams/eq node]
        {\@eqnswtrue\make@display@tag\ltx@label{#1}}}}}}
\makeatother

\def\conf{MFPS 2025} 	%%Fill in the acronym for your conference (with year)
\volume{5}			%Fill in the ENTICS volume number here
\def\volu{5}			% and here
\def\papno{6}			%Fill in your paper number here
\def\mydoi{16832}
\def\lastname{Chernev \emph{et al.}} %Lastnames appear in the running header 
\def\runtitle{Unambiguous Acceptance of Thin Coalgebras} %Short title appears in the running header on even pages. 
\def\copynames{A. Chernev, C. Cîrstea, H.H. Hansen,}% C. Kupke}    %% Fill in the first initial and last name of the authors
\begin{document}
\begin{frontmatter}
    % Title
    \title{Unambiguous Acceptance of Thin Coalgebras}

    % Acknowledgements
    \thanks[ALL]{C\^{i}rstea and Kupke were funded by a Leverhulme Trust Research Project Grant (RPG-2020-232).}
 
    % Authors
    \author{Anton Chernev\thanksref{a}\thanksref{antonEmail}} 
    \author{Corina Cîrstea\thanksref{b}\thanksref{corinaEmail}}
    \author{Helle Hvid Hansen\thanksref{a}\thanksref{helleEmail}}
    \author{Clemens Kupke \thanksref{c}\thanksref{clemensEmail}}

    % Affiliations 
    \address[a]{University of Groningen\\
    Groningen, Netherlands}
    \address[b]{University of Southampton\\
    Southampton, United Kingdom}
    \address[c]{University of Strathclyde\\
    Glasgow, United Kingdom}  							
    % Emails
    \thanks[antonEmail]{Email: \href{mailto:a.chernev@rug.nl} {\texttt{\normalshape
        a.chernev@rug.nl}}}
    \thanks[corinaEmail]{Email: \href{mailto:cc2@ecs.soton.ac.uk} {\texttt{\normalshape
        cc2@ecs.soton.ac.uk}}}
    \thanks[helleEmail]{Email: \href{mailto:h.h.hansen@rug.nl} {\texttt{\normalshape
        h.h.hansen@rug.nl}}}
    \thanks[clemensEmail]{Email: \href{mailto:clemens.kupke@strath.ac.uk} {\texttt{\normalshape
        clemens.kupke@strath.ac.uk}}}
        
    \begin{abstract} 
        Automata admitting at most one accepting run per structure, known as unambiguous automata, find applications in verification of reactive systems as they extend the class of deterministic automata whilst maintaining some of their desirable properties. In this paper, we generalise a classical construction of unambiguous automata from thin trees to thin coalgebras for analytic functors. This achieves two goals: extending the existing construction to a larger class of structures, and providing conceptual clarity and parametricity to the construction by formalising it in the coalgebraic framework. As part of the construction, we link automaton acceptance of languages of thin coalgebras to language recognition via so-called coherent algebras, which were previously introduced for studying thin coalgebras. This link also allows us to establish an automata-theoretic characterisation of languages recognised by finite coherent algebras.
    \end{abstract}


\end{frontmatter}

\section{Introduction}\label{sec:introduction}
\textbf{Background and Motivation}
Model checking~\cite{BaierKatoenPrinciplesOfModelChecking} of reactive systems makes extensive use of automata over infinite objects~\cite{GraedelEtAl}.
A core result facilitating the use of infinite word automata in verification is the determinisation of parity automata. However, this result is limited to infinite words, so systems whose runs exhibit tree-like
structure call for more refined theoretical tools.

Recent work \cite{cirstea_et_al:LIPIcs.CALCO.2017.7,CirsteaKupke:CSL2023} (building on \cite{urabe_et_al:LIPIcs.CONCUR.2016.24}) presents coalgebraic approaches to quantitative model checking using parity automata. Coalgebra~\cite{Rutten:TCS2000} allows for a unified treatment of various system types by viewing these as coalgebras for a functor describing the system type. 
In particular, \cite{CirsteaKupke:CSL2023} proposes an approach to quantitative model checking
of systems with quantitative branching type given by a monad $T$ and the structure of system runs given by a polynomial functor $F$. 
A key condition in \cite{CirsteaKupke:CSL2023} 
is that the property to be checked must be given by an $F$-coalgebra automaton~\cite{KupkeVenema2008CoalgebraicAutomata} that is \emph{unambiguous}, i.e., there is at most one accepting run on each coalgebra.
This raises the question of when an equivalent unambiguous $F$-coalgebra automaton can be constructed from a nondeterministic one.
This question is also of fundamental interest 
and the coalgebraic framework allows to investigate for which system types unambiguous acceptance results can be obtained.
For ordered ranked trees, there are languages that are not accepted by an unambiguous automaton~\cite{CarayolLNW}.
However, for the subclass of \emph{thin trees}, i.e., trees with only countably many infinite branches, \cite{Skrzypczak2016} shows how to construct
from a nondeterministic automaton, an automaton that unambiguously accepts the same thin trees. The construction goes via \emph{thin algebras}: every automaton can be transformed into a finite thin algebra, which can be transformed into an automaton that is unambiguous on thin trees.

Inspired by these results on thin trees, we showed in \cite{ChernevCirsteaHansenKupkeThinCoalg} 
that thin trees and their inductive characterisation can be generalised to the level of $F$-coalgebras for an analytic functor $F$.
Analytic functors~\cite{JoyalFoncteursAnalytiques} include polynomial functors (the type of ordered ranked trees) and quotients thereof such as the bag functor. 
In the present paper, we build on the algebraic characterisation from \cite{ChernevCirsteaHansenKupkeThinCoalg} 
of thin $F$-coalgebras via so-called \emph{coherent algebras}
in order to prove unambiguous acceptance for thin $F$-coalgebras for analytic $F$.

\textbf{Contributions}
We summarise our contributions below.
\begin{itemize}
    \item We show that, when restricting to thin $F$-coalgebras for analytic functors $F$, every (nondeterministic)
    $F$-coalgebra automaton can be transformed into an equivalent unambiguous $F$-coalgebra automaton. 
    We thus extend the results for thin trees~\cite{Skrzypczak2016},
    thereby making a step towards applications in quantitative model checking~\cite{CirsteaKupke:CSL2023}.
    \item We give an automata-theoretic characterisation of languages recognised by finite coherent algebras; 
    these are precisely the languages accepted by $F$-automata with a so-called \emph{prefix-agnostic} acceptance condition, which informally means that acceptance does not depend on any finite prefix of paths in the run.
    \item When instantiated to a polynomial functor $F$, our unambiguous automaton construction provides a categorical account of the classical construction in \cite{Skrzypczak2016}. In particular, thin algebras arise as coherent algebras with additional structure, which we call \emph{rational coherent algebras}.
\end{itemize}

We obtain these results as follows. In order to define unambiguous acceptance, in Section~\ref{sec:coalgebraicAutomata}, we generalise the concept of run for $F$-coalgebra automata in \cite{CirsteaKupke:CSL2023} from polynomial $F$ to analytic $F$.
In Section~\ref{sec:automatonToAlgebra}, we show how to transform an automaton into a finite coherent algebra recognising the same language restricted to thin coalgebras. This construction works not just for parity automata, but, more generally, for automata with a prefix-agnostic acceptance condition.  
We identify \emph{rational coherent algebras} as the coherent algebras obtained from parity automata.
In Section~\ref{sec:algebraToAutomaton}, we show how to transform a finite coherent algebra into an automaton, 
called the \emph{algebraic automaton}, which unambiguously accepts precisely the thin coalgebras that are recognised by the algebra. 
In order to prove correctness of this construction, we show that runs of the algebraic automaton correspond to certain coalgebra-to-algebra morphisms called \emph{markings}.
The uniqueness of markings, and hence of runs, follows from thin coalgebras being recursive
thanks to their inductive structure.
Finally, in Section \ref{sec:mainResults}, we combine the two constructions to obtain our main result, the transformation of an automaton into an automaton which, over thin coalgebras, is unambiguous and equivalent to the original one. 
In addition, we show that the languages recognised by finite coherent algebras coincide with the languages of thin behaviours accepted by automata with prefix-agnostic acceptance.

We finish the section with a brief example of the significance of our unambiguous automaton construction for model checking. Figure \ref{fig:probabilisticServer} depicts (a variant of) the simple probabilistic server from \cite{CirsteaKupke:CSL2023}. The state diagram on the left consists of a server and a worker. At each step, the server process spawns a worker with probability $\frac 1 5$ and returns to itself. A worker process performs a computation with probability $\frac 7 8$ and finishes otherwise. The type of this system is given by the functor $T \circ F$ where $T$ is the distribution monad and $F$ is a polynomial functor with a binary operation $\mathit{fork}$, two unary operations $\mathit{wait}, \mathit{compute}$ and a nullary operation $\mathit{done}$. On the right we see a possible execution (or trace) of the system. Suppose we are given a property $P$ of system executions, such as ``there exists a worker that never finishes''. The framework~\cite{CirsteaKupke:CSL2023} can then determine the probability with which $P$ holds, as long as $P$ is specified by an automaton that has at most one accepting run on each possible execution. Consider the automaton $\mathcal A$ for $P$ that guesses at each $\mathit{fork}$ whether the worker does not terminate. This automaton is ambiguous, so it cannot be readily used for determining the probability.  Note, however, that all executions of the given system are thin, because the server can spawn at most countably many workers. Therefore we can apply our unambiguous automaton construction to $\mathcal A$, so that the resulting automaton satisfies the desired condition of having at most one accepting run for each possible execution.

\begin{figure}
    \centering
    \begin{subfigure}[b]{0.35\textwidth}
        \centering
        \scalebox{0.9}{
            \begin{tikzcd}[ampersand replacement=\&]
                {\text{server}} \& \bullet \& {\text{worker}} \& {}
                \arrow["{\frac{4}{5}, wait}", from=1-1, to=1-1, loop, in=55, out=125, distance=10mm]
                \arrow["{\frac{1}{5}, fork}", from=1-1, to=1-2]
                \arrow["old", curve={height=-12pt}, from=1-2, to=1-1]
                \arrow["new", from=1-2, to=1-3]
                \arrow["{\frac 7 8, compute}", from=1-3, to=1-3, loop, in=55, out=125, distance=10mm]
                \arrow["{\frac 1 8, done}", from=1-3, to=1-4]
            \end{tikzcd}
        }
        \caption{State diagram}
    \end{subfigure}%
    \begin{subfigure}[b]{0.65\textwidth}
        \centering
        \scalebox{0.9}{
            \begin{tikzcd}[ampersand replacement=\&]
                {\text{server}} \& \bullet \& {\text{server}} \& \bullet \& {\text{server}} \& \dotsb \\
                \& {\text{worker}} \& {} \& {\text{worker}} \& {\text{worker}} \& \dotsb
                \arrow["fork", from=1-1, to=1-2]
                \arrow["old", from=1-2, to=1-3]
                \arrow["new"', from=1-2, to=2-2]
                \arrow["fork", from=1-3, to=1-4]
                \arrow["old", from=1-4, to=1-5]
                \arrow["new"', from=1-4, to=2-4]
                \arrow["wait", from=1-5, to=1-6]
                \arrow["done", from=2-2, to=2-3]
                \arrow["compute", from=2-4, to=2-5]
                \arrow["compute", from=2-5, to=2-6]
            \end{tikzcd}
        }
        \caption{Possible execution (trace)}
    \end{subfigure}
    \caption{Example of a probabilistic server}
    \label{fig:probabilisticServer}
\end{figure}

\section{Preliminaries}\label{sec:preliminaries}
\subsection{Automata and Algebras for Languages of Infinite Words}
We begin by reviewing basics from the classical theory of automata on infinite words~\cite{GraedelEtAl}. There are multiple types of equivalent infinite word automata, but here we focus on \emph{(nondeterministic) parity word automata}. Given a finite \emph{alphabet} $\Sigma$, a nondeterministic parity word automaton is a tuple $\mathcal A = (Q, \delta, Q_I, \Omega)$, where $Q$ is a finite set of states, $\delta: Q \to \Pow(\Sigma \times Q)$ is a transition function, $Q_I \subseteq Q$ is a set of initial states and $\Omega: Q \to \omega$ is a priority function. An \emph{accepting run} of $\mathcal A$ on an infinite word $x = (a_n)_{n\in\omega} \in \Sigma^\omega$ is a sequence of states $(q_n)_{n\in\omega} \in Q^\omega$ such that $q_0 \in Q_I$, $(a_n, q_{n+1}) \in \delta(q_n)$ and $\limsup_{n\in\omega}\Omega(q_n)$ is even, i.e., the largest priority occurring infinitely often is even. An infinite word $x$ is \emph{accepted} by $\mathcal A$ if there exists an accepting run of $\mathcal A$ on $x$. Languages (i.e., sets) of infinite words accepted by a nondeterministic parity word automaton are called \emph{$\omega$-regular}. An automaton $\mathcal A$ is \emph{deterministic} if $Q_I$ is a singleton and for each $q \in Q$ and $a \in \Sigma$, we have a single pair $(a, q_1) \in \delta(q)$. For convenience, we write deterministic parity word automata as $(Q, \delta, q_I, \Omega)$ where $\delta: Q \to Q^\Sigma$ and $q_I \in Q$. An important result is that deterministic parity word automata accept the same languages as all (nondeterministic) parity word automata.

An alternative, algebraic approach to characterising $\omega$-regular languages is via \emph{$\omega$-semigroups}~\cite[Chapter~2]{PerrinPin:2004:InfiniteWords}. An $\omega$-semigroup is a two-sorted algebraic structure $(V, W)$ with three operations $\cdot: V \times V \to V$, $\times: V \times W \to W$, $\Pi: V^\omega \to W$, satisfying certain associativity axioms. In order to get some intuition about $\omega$-semigroups, consider $(\Sigma^+, \Sigma^\omega)$, which is the $\omega$-semigroup freely generated by $\Sigma$. Here $\cdot$ is concatenation between two finite words, $\times$ is concatenation between a finite and an infinite word and $\Pi$ is concatenation of infinitely many finite words. A \emph{homomorphism} between $\omega$-semigroups $(V_1, W_1)$ and $(V_2, W_2)$ is a pair of maps $f = (f_V, f_W)$, where $f_V: V_1 \to V_2$, $f_W: W_1 \to W_2$, that preserves the $\omega$-semigroup operations. The key property of $\omega$-semigroups is that $L \subseteq \Sigma^\omega$ is $\omega$-regular if and only if there exists a finite $\omega$-semigroup $(V, W)$, a homomorphism $f: (\Sigma^+, \Sigma^\omega) \to (V, W)$ and a \emph{recognising set} $U \subseteq W$ such that $L = f_W^{-1}(U)$.

There exist extensions of parity automata from words to other infinite structures, such as binary trees. Instead of considering automata running on some concrete  structures, we will work with \emph{$F$-coalgebra automata} (see Section \ref{sec:coalgebraicAutomata}) that run on \emph{$F$-coalgebras}.

\subsection{F-Coalgebras and F-Algebras}
\emph{$F$-coalgebras}~\cite{Rutten:TCS2000} are a formalism for modelling state-based systems that is parametric in the transition type $F$. Let $F$ be an endofunctor on the category $\Set$. An $F$-coalgebra is a tuple $(X, \xi)$ consisting of an object $X$ and a morphism $\xi: X \to FX$. An \emph{$F$-coalgebra morphism} $f: (X,\xi) \to (Y,\upsilon)$ is a map $f: X \to Y$ (in $\Set$) such that $\upsilon \circ f = Ff \circ \xi$. Informally, $F$-coalgebra morphisms map states in such a way that the transition structure is preserved. $F$-coalgebras, together with $F$-coalgebra morphisms, form a category. A terminal object $(Z, \zeta)$ in this category is called a \emph{final $F$-coalgebra} and its elements can be thought of as \emph{abstract behaviours}. By selecting a \emph{root state} $x_I$ in a coalgebra $(X, \xi)$, we get a \emph{pointed $F$-coalgebra} $(X,\xi,x_I)$. \emph{Pointed $F$-coalgebra morphisms} are $F$-coalgebra morphisms that also preserve the root.

Given F-coalgebras $(X,\xi)$ and $(Y,\upsilon)$, two states $x \in X$ and $y \in Y$ are \emph{behaviourally equivalent} if there exist  $F$-coalgebra morphisms $f_X: (X,\xi) \to (W, \eta)$ and $f_Y: (Y,\upsilon) \to (W,\eta)$ into a third $F$-coalgebra such that $f_X(x) = f_Y(y)$.
Two pointed $F$-coalgebras $(X,\xi,x_I)$, and $(Y,\upsilon, y_I)$ are behaviourally equivalent if $x_I$ and $y_I$ are behaviourally equivalent. Under the assumption that $F$ preserves weak pullbacks, behavioural equivalence amounts to the existence of a \emph{span} of pointed coalgebra morphisms, i.e., a pointed $F$-coalgebra $(R, \rho, s_I)$ with pointed morphisms $f_X: (R,\rho, s_I) \to (X,\xi, x_I)$ and $f_Y: (R, \rho, s_I) \to (Y, \upsilon, y_I)$.

Assuming that $F$ preserves intersections and preimages, there exists a natural transformation $\Base_F: F \Rightarrow \Pow$, where $\Pow$ is the covariant power-set functor (see \cite[Theorem~8.1]{Gumm2005FromTCoalgebrasToFilterStructures}). For $\bar x \in FX$, $\Base_F(\bar x) \subseteq X$ is the least set such that $\bar x \in F(\Base_F(\bar x))$. The notion of base allows us to define \emph{reachable} pointed $F$-coalgebras. These are pointed coalgebras $(X,\xi,x_I)$ where for every $x \in X$, there exists a finite sequence $x_0, x_1, \dotsc, x_n$ such that $x_0 = x_I$, $x_n = x$ and $x_{i+1} \in \Base_F(\xi(x_i))$ for all $i < n$. As the name suggests, every state in a reachable coalgebra can be reached from the root along some transitions. One readily observes that reachable coalgebras come with an induction principle: if $P \subseteq X$ is a property such that $x_I \in P$ and, for all $x \in X$, $x \in P$ implies $\Base(\xi(x)) \subseteq P$, then $P = X$.

\emph{$F$-algebra} is the dual notion of $F$-coalgebra. An $F$-algebra is a pair $(C,\gamma)$ with $\gamma: FC \to C$. An \emph{$F$-algebra morphism} $f: (B,\beta) \to (C,\gamma)$ is then a map $f: B \to C$ with $f \circ \beta = \gamma \circ Ff$. An \emph{initial $F$-algebra} is an initial object in the category of $F$-algebras and $F$-algebra morphisms. An $F$-algebra can be thought of as an algebra with a (generalised) signature $F$, and the elements of an initial $F$-algebra can be seen as terms over this signature.

Given an $F$-coalgebra $(X,\xi)$ and an $F$-algebra $(C,\gamma)$, an \emph{$F$-coalgebra-to-algebra morphism} is a map $f: X \to C$ satisfying $f = \gamma \circ Ff \circ \xi$. An $F$-coalgebra $(X,\xi)$ is \emph{recursive} if for every $F$-algebra $(C, \gamma)$, there exists a unique $F$-coalgebra-to-algebra morphism from $(X,\xi)$ to $(C,\gamma)$. Recursive coalgebras capture the idea of recursion on well-founded relations (see \cite{AdamekEtAlWellFoundedRecursiveCoalgebras} for details).

In this paper, we will work with coalgebras for \emph{analytic functors}.

\subsection{Analytic Functors}
\emph{Analytic functors}~\cite{JoyalFoncteursAnalytiques} (see also \cite{HasegawaAnalyticFunctors}) generalise polynomial functors by allowing 
symmetries of successors, thus including, for instance, the bag functor. They were shown in \cite{ChernevCirsteaHansenKupkeThinCoalg} to be a natural setting for studying \emph{thin coalgebras} (see Section \ref{subsec:thinCoalgebras}). While here we give the basic definitions, we refer the reader to \cite[Sections~II,III]{ChernevCirsteaHansenKupkeThinCoalg} for a more detailed discussion with examples.

Given sets $X, U$ and a group $H$ of permutations on $U$, $H$ \emph{acts} on the set $X^U$ of functions by $\sigma \cdot \phi = \phi \circ \sigma^{-1}$, for $\sigma \in H$ and $\phi \in X^U$. The set of \emph{orbits} of this action is written as $X^U / H$, with elements of the form $[\phi]_H = \{ \psi \in X^U \mid \exists \sigma \in H(\psi = \sigma \cdot \phi) \}$. An analytic functor is a functor of the form $F(X) = \bigsqcup_{i \in I} X^{U_i} / H_i$ where $I$ is an index set, $U_i$ is a finite set and $H_i$ is a group of permutations on $U_i$, for all $i \in I$. Thus elements of $F(X)$ are of the form $(i, [\phi]_{H_i})$. We think of the sets $U_i$ as positions to which we assign data in $X$. These positions can be permuted according to $H_i$.
For a function $f: X \to Y$, $F(f)(i,[\phi]_{H_i}) = (i, [f \circ \phi]_{H_i})$. 

We will use the notion of \emph{functor derivative}~\cite{AbbottEtAl:DifferentiatingDataStructures} for an analytic functor $F$, which models \emph{one-hole contexts} over $F$. Consider the collection of functions $\bigsqcup_{u \in U} X^{U \setminus \{ u\}}$, which can be seen as the collection of partial functions from $U$ to $X$ that are undefined precisely at one element. A group $H$ of permutations on $U$ acts on $\bigsqcup_{u \in U} X^{U \setminus \{u\}}$ by $\sigma \cdot (u, \phi) = (\sigma(u), \phi \circ (\sigma^{-1}|_{U \setminus \{\sigma(u)\}}))$, for $u \in U$, $\phi: U \setminus \{ u \} \to X$. The orbit of an element $(u, \phi) \in  \bigsqcup_{u \in U} X^{U \setminus \{u\}}$ is denoted by $[u, \phi]_{H_i}$. The functor derivative of $F$ is the functor $F'(X) = \bigsqcup_{i \in I}\big(\bigsqcup_{u \in U_i}X^{U_i \setminus \{u\}}\big) / H_i$. Elements of $F'(X)$ are of the form $(i, [u, \phi]_{H_i})$ and are called one-hole contexts, because one position is empty. An element $x \in X$ can be ``plugged'' into a context $(i,[u, \phi]_{H_i}) \in F'X$, resulting in $(i,[\phi \cup \{ (u, x) \}]_{H_i}) \in FX$. Formally, define the \emph{context plug-in} natural transformation $\trig: F' \times \Id \Rightarrow F$ by $\trig_X((i,[u, \phi]_{H_i}), x) \coloneqq (i,[\phi \cup \{ (u, x) \}]_{H_i})$.

\begin{proposition}
\label{prop:trigNaturalitySquareIsWeakPullback}
    The plug-in is weakly cartesian, i.e., every naturality square of $\trig$ is a weak pullback.
\end{proposition}

We often use the following notational convention: given a set $X$, write $x \in X$, $\bar x \in FX$ and $\bar x' \in F'X$. 

Analytic functors and their derivatives satisfy the conditions for the existence of a base. Concretely, their base is given by $\Base_F([\phi]_{H_i}) = \Im(\phi)$ and $\Base_{F'}([u, \phi]) = \Im(\phi)$ for $\phi \in FX$, $[u, \phi]_{H_i} \in F'X$. We have the property $\Base_F(\trig_X(\bar x', x)) = \Base_{F'}(\bar x') \cup \{x\}$, for $\bar x' \in F'X$ and $x \in X$. Moreover, if $\bar x \in FX$ and $x \in \Base_F(\bar x)$, there exists a (not necessarily unique) $\bar x' \in F'X$ with $\trig_X(\bar x', x) = \bar x$ (see \cite{ChernevCirsteaHansenKupkeThinCoalg} for details). Analytic functors also preserve weak pullbacks.

We introduce a new \emph{context decomposition} natural transformation. Intuitively, context decomposition $\decomp: F \Rightarrow F(F' \times \Id)$ does the opposite of context plug-in: it gives all possible ways to split $\bar x \in FX$ into a context in $\bar x' \in F'X$ and an element in $x \in X$. Moreover, it organises all decompositions $(\bar x', x)$ of $\bar x$ into an $F$-structure, based on the position of the context hole. For each such decomposition $(\bar x', x)$, think of $\bar x'$ as the context of siblings of $x$ in $\bar x$. This will be essential in Definition \ref{def:algebraicAutomaton} (the algebraic automaton).

\begin{definition}
\label{def:decomp}
    Given an analytic functor $F = \bigsqcup_{i \in I} X^{U_i} / H_i$, define the \emph{context decomposition} natural transformation $\decomp: F \Rightarrow F(F' \times \Id)$ as follows:
    \begin{equation*}
        \decomp_X(i, [\phi]_{H_i}) \coloneqq (i, [\psi]_{H_i}), \quad \text{where} \quad 
        \psi: U_i \to F'X \times X, \quad
        \psi(u) \coloneqq ((i, [u, \phi \setminus \{ \langle u, \phi(u) \rangle \} ]_{H_i}), \phi(u)).
    \end{equation*}
\end{definition}

\begin{example}
    Take $F(X) = \Sigma \times X^3$ where $\Sigma = \{a, b\}$. Then $F'(X) = \Sigma \times 3 \times X^2$, and for $X = \{x_0, x_1, x_2\}$ and $(a, x_0, x_1, x_2) \in FX$, we have:
    \begin{equation*}
        \decomp_X(a, x_0, x_1, x_2) = (a, ((a, 0, x_1, x_2), x_0), ((a, 1, x_0, x_2), x_1), ((a, 2, x_0, x_1), x_2)).
    \end{equation*}
\end{example}

\begin{example}
    Take $F = \Bag_3$, the bag functor where the bag size is bounded by $3$, i.e., $F = \bigsqcup_{n \leq 3} X^{n} /H$ where $H$ is the symmetric group on $n$. Then $F' = \Bag_2$. We use the notation $\{ \dotsc \}_b$ for bags. For $X = \{ x_0, x_1, x_2 \}$, we have:
    \begin{equation*}
        \decomp_X(\{x_0, x_0, x_1 \}_b) = \{ (\{x_0, x_1 \}_b, x_0), (\{ x_0, x_1\}_b,x_0), (\{x_0, x_0\}_b, x_1) \}_b.
    \end{equation*}
\end{example}

In order to avoid working with the concrete definition of $\decomp$, we identify its key abstract properties. Below we write $\prj{1}$ and $\prj{2}$ for product projections (later, we also write $\inj{1}$ and $\inj{2}$ for coproduct injections).

\begin{lemma}
\label{lem:decompProperties}
    Context decomposition $\decomp$ satisfies:
    \begin{enumerate}
        \item[(i)] $F\prj{2} \circ \decomp_X = \id$ (see Figure \ref{fig:decompProp1});
        \item[(ii)] for every element $\upsilon = (\bar y', y) \in F'(F'X \times X) \times (F'X \times X)$ with $\trig_{F'X \times X}(\upsilon) \in \decomp_X[FX]$, we have $F'\prj{2}(\bar y') = \prj{1}(y)$ (see Figure \ref{fig:decompProp2});
        \item[(iii)] if $\bar x \in FX$ and $(\bar x', x) \in \Base_{F}(\decomp_X(\bar x))$, then $\trig_X(\bar x', x)  = \bar x$.
    \end{enumerate}
    \begin{figure}[ht]
        \centering
        \begin{subfigure}[b]{0.5\textwidth}
            \centering
            \begin{tikzcd}
                FX & {F(F'X \times X)} \\
                & FX
                \arrow["\decomp_X", from=1-1, to=1-2]
                \arrow["\id"', from=1-1, to=2-2]
                \arrow["{F\prj{2}}", from=1-2, to=2-2]
            \end{tikzcd}
            \caption{Property (i)}
            \label{fig:decompProp1}
        \end{subfigure}%
        \begin{subfigure}[b]{0.5\textwidth}
            \centering
            \begin{tikzcd}
                1 & {F'(F'X\times X)} \\
                {F'X\times X} & {F'X}
                \arrow["{\prj{1} \circ \upsilon}", from=1-1, to=1-2]
                \arrow["{\prj{2} \circ \upsilon}"', from=1-1, to=2-1]
                \arrow["{F'\prj{2}}", from=1-2, to=2-2]
                \arrow["{\prj{1}}"', from=2-1, to=2-2]
            \end{tikzcd}
            \caption{Property (ii)}
            \label{fig:decompProp2}
        \end{subfigure}
        \caption{Diagrams for Lemma \ref{lem:decompProperties}.}
        \label{fig:decompProperties}
    \end{figure}
\end{lemma}

Property (i) completely describes the content of the $\Id$-component of $\decomp_X(\bar x) \in F(F' \times \Id)$. Together with property (i), property (ii) completely describes the $F'$-component. Thus these two properties can be taken as an abstract, equivalent definition of $\decomp$. Property (iii) follows from (i) and (ii) and it conveys our intuitive understanding that $\decomp$ decomposes $\bar x \in FX$ into pairs of an element $x$ and its siblings $\bar x'$.

We apply the concept of \emph{relation lifting}~\cite{KurzVelebilRelationLifting} for analytic functors. Specifically, we will use the lifting $\inlift$ of the ``element of'' relation $\in$.
Given a set $X$, $p \in FX$ and $q \in F(\Pow(X))$, we have $p \inlift q$ if there exists $r \in F({\in})$ such that $F\prj{1}(r) = p$ and $F\prj{2}(r) = q$. Informally, $p \inlift q$ means ``$p$ and $q$ have matching indices in $I$ and $p$ is position-wise contained in $q$''. The parameters $F$ and $X$, on which $\inlift$ depends, are left implicit and understood from the context.

\textbf{Assumption}. For the rest of the paper, we fix an analytic functor $FX = \bigsqcup_{i \in I} X^{U_i} / H_i$ where $I$ is finite. This ensures that $F$ and $F'$ preserve finite sets.

\subsection{Thin Coalgebras}
\label{subsec:thinCoalgebras}
We are interested in running $F$-automata on a subclass of $F$-coalgebras called \emph{thin $F$-coalgebras}~\cite{ChernevCirsteaHansenKupkeThinCoalg}. Thin coalgebras generalise the notion of \emph{thin tree}~\cite{Skrzypczak2016} to the level of coalgebras. They are defined as those $F$-coalgebras for which every state is the starting point of only countably many infinite paths. More precisely, given a $F$-coalgebra $(X,\xi)$ and $x \in X$ with $\xi(x) = (i, [\phi]_{H_i})$, we say that an element $x_1 \in \Base_F(i, [\phi]_{H_i}) = \Im(\phi)$ is a \emph{successor} of $x$ with multiplicity $|\phi^{-1}(x_1)|$. The successor relation on $(X,\xi)$ defines a multigraph, with multiplicities corresponding to multiple parallel edges. A state $x \in X$ is \emph{thin} if there are only countably many infinite paths starting from $x$ in this multigraph. A (pointed) coalgebra is \emph{thin} if all its states are thin.

Behaviours of thin coalgebras can be characterised algebraically via \emph{coherent $(F+G)$-algebras}. Define the functor $G(X) \coloneqq (F'X)^\omega$, mapping $X$ to the set of streams of contexts over $X$. An $(F+G)$-algebra is of the form $(C, \gamma)$, with $\gamma = [\gamma_0, \gamma_1]$, where $\gamma_0: FX \to X$ is an $F$-algebra structure and $\gamma_1: GX \to X$ is a $G$-algebra structure. An $(F+G)$-algebra $(C,\gamma)$ is \emph{coherent} if it satisfies the equation $\gamma_1 = \gamma_0 \circ \trig_C \circ \langle \id, \gamma_1 \rangle \circ \langle \head, \tail \rangle$, where $\head$ stands for stream head and $\tail$ stands for stream tail. Roughly, the equation says ``evaluating a stream with $\gamma_1$ is equal to evaluating the stream tail, plugged into the stream head, with $\gamma_0$''. \cite[Corollary~VII.6]{ChernevCirsteaHansenKupkeThinCoalg} shows that the initial coherent $(F+G)$-algebra is isomorphic to the collection of behaviours of thin coalgebras.

The initial coherent $(F+G)$-algebra is given concretely as follows. Fix an initial $(F+G)$-algebra $(A, \alpha = [\alpha_0, \alpha_1])$ and a final $F$-coalgebra $(Z,\zeta)$ (their existence is proven in \cite{ChernevCirsteaHansenKupkeThinCoalg}). There exists a natural way to interpret terms $a \in A$ in $Z$. Informally speaking, for $\bar a \in FA$, $a = \alpha_0(\bar a)$ is interpreted as a state with successors $\bar a$; for $(\bar a'_n)_{n\in\omega} \in GA$, $a = \alpha_1((\bar a'_n)_{n \in \omega})$ is interpreted by successively plugging all contexts $(\bar a'_n)_{n\in\omega}$ into each other, i.e., plugging $\bar a'_1$ into $\bar a'_0$, $\bar a'_2$ into $\bar a'_1$, $\bar a'_3$ into $\bar a'_2$ and so on. This is formalised by defining a suitable $(F+G)$-algebra structure $\beta = [\beta_0,\beta_1]$ on $Z$ and taking the \emph{semantics map} $\interpr-: (A, \alpha) \to (Z, \beta)$ to be the unique map obtained by initiality of $(A,\alpha)$. By taking the image $\thin Z \subseteq Z$\footnote{The superscript ${\mathit{\text{\th}}}$ is pronounced as ``thin''. The letter \emph{thorn} \text{\th} denotes a dental fricative (e.g., the first sound in ``thin'') in Old English~\cite{Jekiel2012DentalFricatives}.} of the semantics map, one obtains both an $(F+G)$-subalgebra $(\thin Z, \thin \beta)$ of $(Z,\beta)$ and an $F$-subcoalgebra $(\thin Z, \thin \zeta)$ of $(Z, \zeta)$. 
Figure \ref{fig:initialCoherentAlg} gives a visual summary.
\begin{figure}
    \centering
    \begin{tikzcd}
        {(F+G)A} & {(F+G)\thin Z} & {(F+G)Z} \\
        A & {\thin Z} & Z \\
        & {F\thin Z} & FZ
        \arrow["{{(F+G)\interpr-}}", from=1-1, to=1-2]
        \arrow["\alpha"', from=1-1, to=2-1]
        \arrow[from=1-2, to=1-3]
        \arrow["{{\thin \beta}}"', from=1-2, to=2-2]
        \arrow["\beta", from=1-3, to=2-3]
        \arrow["{\interpr-}", two heads, from=2-1, to=2-2]
        \arrow[hook, from=2-2, to=2-3]
        \arrow["{{\thin \zeta}}"', from=2-2, to=3-2]
        \arrow["\zeta", from=2-3, to=3-3]
        \arrow[from=3-2, to=3-3]
    \end{tikzcd}
    \caption{Algebra and coalgebra on $\thin Z$.}
    \label{fig:initialCoherentAlg}
\end{figure}
We have that $(\thin Z, \thin \beta)$ is an initial coherent $(F+G)$-algebra, i.e., for every coherent $(F+G)$-algebra $(C,\gamma)$, there exists a unique $(F+G)$-algebra morphism $\cev_{(C,\gamma)}: (\thin Z, \thin \beta) \to (C, \gamma)$. Moreover, $(\thin Z, \thin \zeta)$ is a final thin coalgebra, meaning that for every thin coalgebra $(X, \xi)$, there exists a unique $F$-coalgebra morphism $\tbeh_{(X,\xi)}: (X,\xi) \to (\thin Z, \thin \zeta)$. In other words, $(\thin Z, \thin \zeta)$ is the subcoalgebra of all \emph{thin behaviours}, i.e., behaviours of thin coalgebras.  Furthermore, $\thin Z$ is isomorphic to the collection of \emph{normal terms}~\cite[Section~V]{ChernevCirsteaHansenKupkeThinCoalg}: each $z \in \thin Z$ has a canonical \emph{normal representative} $a \in A$ with $\interpr{a} = z$. We have two useful properties connecting $\thin \beta$ and $\thin \zeta$:
\begin{align}
    &\thin \zeta = (\thin \beta_0)^{-1}, \label{eq:betaZetaProp1} \\
    &(z_m)_{m\in\omega} \in (\thin Z)^\omega, (\bar z'_{m})_{m > 0} \in (F'\thin Z)^\omega, \forall m \in \omega (\trig_{\thin Z}(\bar z'_{m+1},z_{m+1}) = \thin \zeta(z_m)) \Longrightarrow \thin \beta_1((\bar z'_m)_{m > 0}) = z_0. \label{eq:betaZetaProp2}
\end{align}

In the present work, we are interested in the language recognition aspect of $(F+G)$-algebras. Given a coherent $(F+G)$-algebra $(C,\gamma)$ and $U \subseteq C$, the \emph{language} of the triple $(C,\gamma,U)$ is defined as $L(C,\gamma,U) \coloneqq \cev^{-1}_{(C,\gamma)}(U) \subseteq \thin Z$. Hence coherent algebras recognise languages of thin behaviours, similarly to how $\omega$-semigroups recognise languages of infinite words. We refer to such a triple $(C,\gamma,U)$ as a coherent $(F+G)$-algebra with a \emph{recognising set}.

\section{Runs and Unambiguity of F-Coalgebra Automata}
\label{sec:coalgebraicAutomata}
In this section, we present \emph{$F$-coalgebra automata} (for brevity, $F$-automata), which were studied in \cite{KupkeVenema2008CoalgebraicAutomata} as automata accepting $F$-coalgebras. We define acceptance of $F$-automata via the notion of \emph{run}, in contrast with \cite{KupkeVenema2008CoalgebraicAutomata}, which defines acceptance via parity games. Our reason for introducing runs is to be able to define \emph{unambiguous $F$-automata}. While the two definitions of acceptance (via runs and via parity games) appear to coincide, we do not show it in this paper, as we work exclusively with runs. We note that a similar definition of $F$-automaton runs and unambiguity is given in \cite{CirsteaKupke:CSL2023}, but only for polynomial functors $F$.

\begin{definition}
    An \emph{$F$-automaton} is a quadruple $\mathcal{A} = (Q, \delta, Q_I, \Acc)$ where 
    $Q$ is a finite set of \emph{states}, $\delta \colon Q \to (\Pow \circ F)(Q)$ is a \emph{transition function}.
    $Q_I \subseteq Q$ is a set of \emph{initial states}, and $\Acc \subseteq Q^\omega$ is an \emph{acceptance condition}.
\end{definition}

According to the above definition, $F$-automata are, in general, nondeterministic, i.e., every state $q \in Q$ has an arbitrary set $\delta(q)$ of transitions and there are multiple initial states $Q_I$. We do not put any restrictions on the acceptance condition; instead, we distinguish the following types of acceptance conditions.

\begin{definition}
    Let $\mathcal A = (Q, \delta, Q_I, \Acc)$ be an $F$-automaton. We call $\Acc$:
    \begin{itemize}
        \item \emph{parity} if there exists a map $\Omega: Q \to \mathbb N$ such that $(q_n)_{n \in \omega} \in \Acc$ if and only if $\limsup_{n \in \omega} \Omega(q_n)$ is even;
        \item \emph{$\omega$-regular} if $\Acc$ is an $\omega$-regular language over the alphabet $Q$;
        \item \emph{prefix-agnostic} if for all $x \in Q^\omega$, $w \in Q^*$: $wx \in \Acc$ if and only if $x \in \Acc$.
    \end{itemize}
\end{definition}

$F$-automata with a parity acceptance condition are known as \emph{parity $F$-automata} and we write them as a tuple $(Q, \delta, Q_I, \Omega)$, with $\Omega$ instead of $\Acc$. By taking the polynomial functor $F(X) = \Sigma \times X$, for some alphabet $\Sigma$, we obtain nondeterministic parity word automata.

Since parity word automata recognise $\omega$-regular languages, one can see that every parity condition is also $\omega$-regular. Conversely, every $F$-automaton with $\omega$-regular acceptance can be turned into an equivalent parity $F$-automaton via the \emph{wreath product} construction~\cite[Theorem~4.4]{KupkeVenema2008CoalgebraicAutomata}.
Parity conditions are also prefix-agnostic, but automata with the prefix-agnostic conditions turn out to be strictly more expressive, as shown below.

\begin{example}
\label{ex:automatonNonRegularLanguage}
    Consider the functor $F(X) = \Sigma \times X$, for $\Sigma = \{ a, b \}$, whose derivative is $F'(X) = \Sigma$. Define the (word) $F$-automaton $\mathcal A = (Q, \delta, Q_I, \Acc)$ with $Q \coloneqq \{ q_a, q_b \}$, $\delta(q) \coloneqq \{ (a, q_a), (b, q_b) \}$ for all $q \in Q$, $Q_I \coloneqq Q$. Let $\Acc \subseteq Q^\omega$ consist of those infinite words that contain infinitely many $q_b$'s and unboundedly many consecutive $q_a$'s, i.e., for every natural number $n$, the word contains $n$-many consecutive $q_a$'s. One readily sees that $\Acc$ is prefix-agnostic and $\mathcal A$ accepts (in the classical sense) the language $L \coloneqq \{ a_0 a_1 \dotsc \mid q_{a_0} q_{a_1} \dotsc \in \Acc\}$. However, $L$ is not $\omega$-regular. This is because every non-empty $\omega$-regular language contains an ultimately periodic word, i.e., a word of the form $wu^\omega$, while $L$ contains no such words. This example shows that automata with prefix-agnostic acceptance are more expressive than parity automata.
\end{example}

Next, we define $F$-automaton runs and unambiguity, thereby generalising the definitions in \cite{CirsteaKupke:CSL2023} from polynomial functors to arbitrary analytic functors. Below we write $\Delta_Y$ for the constant functor sending every set to the set $Y$ and every function to $\id_Y$.

\begin{definition}
    Let $\mathcal{A} = (Q, \delta, Q_I, \Acc)$ be an $F$-automaton and $(X, \xi, x_I)$ be a pointed $F$-coalgebra. A \emph{pre-run} of $\mathcal A$ on $(X, \xi, x_I)$ is a reachable pointed $(F \times \Delta_X \times \Delta_Q)$-coalgebra $(R, \rho = \langle \rho_F, \rho_X, \rho_Q \rangle, r_I)$ satisfying:
    \begin{enumerate}
        \item[(i)] $\rho_X: (R, \rho_F, r_I) \to (X, \xi, x_I)$ is a pointed $F$-coalgebra morphism;
        \item[(ii)] $(F\rho_Q \circ \rho_F)(r) \in (\delta \circ \rho_Q)(r)$ for all $r \in R$;
        \item[(iii)] $(\rho_Q(r_n))_{n\in\omega} \in \Acc$ for all $(r_n)_{n\in\omega} \in R^\omega$ with $r_0 = r_I$ and $\forall n (r_{n+1} \in \Base_F(\rho_F(r_n))$.
    \end{enumerate}
    We define a \emph{run} as a pre-run $(R,\rho,r_I)$ for which $(R,\rho)$ is a subcoalgebra of the final $(F \times \Delta_X \times \Delta_Q)$-coalgebra. A (pre-)run is \emph{accepting} if $\rho_Q(r_I) \in Q_I$. We say that $\mathcal A$ \emph{accepts} $(X, \xi, x_I)$ if there exists an accepting run of $\mathcal A$ on $(X, \xi, x_I)$.
\end{definition}

\begin{definition}
    Let $\mathcal A$ be an $F$-automaton and $(X, \xi, x_I)$ be a pointed $F$-coalgebra. We say that $\mathcal A$ is \emph{unambiguous} on $(X, \xi, x_I)$ if $\mathcal A$ has at most one accepting run on $(X, \xi, x_I)$.
\end{definition}

A pre-run of $\mathcal A$ on $(X,\xi, x_I)$ represents an execution of $\mathcal A$ on the structure of $(X,\xi,x_I)$. The pre-run yields a span
\begin{tikzcd}
    X & R & Q
    \arrow["{\rho_X}"', from=1-2, to=1-1]
    \arrow["{\rho_Q}", from=1-2, to=1-3]
\end{tikzcd}
where $R$ is equipped with an $F$-coalgebra structure $\rho_F\colon R \to F(R)$.
Property (i) says that $\rho_X\colon R \to X$ respects the coalgebra structure of $X$; 
property (ii) says that $\rho_Q \colon R \to Q$ respects the automaton transitions; and property (iii) says that $\rho_Q$ respects the automaton acceptance condition. We note that, classically, automaton runs need not respect the acceptance condition (they are called \emph{final} if they do) but in this paper, all (pre-)runs are required to be final. Runs have the additional property that behaviourally equivalent
elements are identified. This is essential for the definition of unambiguous automata, where we count the number of accepting runs of an automaton modulo behavioural equivalence. Intuitively, runs are as close as possible to being a relation between $X$ and $Q$. However, there could be elements $r_1 \neq r_2 \in R$ of a run $R$ with $\rho_X(r_1) = \rho_X(r_2)$ and $\rho_Q(r_1) = \rho_Q(r_2)$, since $r_1$ and $r_2$ can still differ on $\rho_F$.

\begin{figure}
    \centering
    \begin{subfigure}[b]{0.20\textwidth}
        \centering
        \begin{tikzcd}
            {x_1} \\
            {x_2} & {x_3}
            \arrow[from=1-1, to=2-1]
            \arrow[from=1-1, to=2-2]
            \arrow[from=2-2, to=2-2, loop, in=50, out=0, distance=7mm]
        \end{tikzcd}
        \caption{Pointed coalgebra}
        \label{fig:pointedCoalgebraExample}
    \end{subfigure}%
    \begin{subfigure}[b]{0.47\textwidth}
        \centering
        \begin{tikzcd}
            {r_1} \\
            {r_2} & {r_3} & {r_4} & {r_5} & \dotsb
            \arrow[from=1-1, to=2-1]
            \arrow[from=1-1, to=2-2]
            \arrow[from=2-2, to=2-3]
            \arrow[from=2-3, to=2-4]
            \arrow[from=2-4, to=2-5]
        \end{tikzcd}
        \caption{Pre-run}
        \label{fig:prerunExample}
    \end{subfigure}%
    \begin{subfigure}[b]{0.32\textwidth}
        \centering
        \begin{tikzcd}
            {r_1} \\
            {r_2} & {r_3} & {r_4} & {r_5}
            \arrow[from=1-1, to=2-1]
            \arrow[from=1-1, to=2-2]
            \arrow[from=2-2, to=2-3]
            \arrow[from=2-3, to=2-4]
            \arrow[curve={height=20pt}, from=2-4, to=2-2]
        \end{tikzcd}
        \caption{Run}
        \label{fig:runExample}
    \end{subfigure}
    \caption{(Pre-)runs from Example \ref{ex:preruns}}
    \label{fig:prerunAndRunExample}
\end{figure}

\begin{example}
\label{ex:preruns}
    Consider the polynomial functor $FX = X^2 + X + 1$. Figure \ref{fig:pointedCoalgebraExample} depicts a pointed $F$-coalgebra $(X, \xi, x_1)$ with $\xi(x_1) = (x_2, x_3)$, $\xi(x_2) = ()$ and $\xi(x_3) = (x_3)$. Consider the $F$-automaton $\mathcal A = (Q, \delta, Q_I, \Acc)$ with
    $Q = \{q_1, q_2, q_3, q_4 \}$, $Q_I \coloneqq \{q_1\}$, $\delta(q_1) = \{ (q_2, q_3) \}$, $\delta(q_2) = \{ () \}$, $\delta(q_3) = \{ (q_3), (q_4) \}$, $\delta(q_4) = \{ (q_3) \}$ and $\Acc$ containing all $(q_n)_{n \in \omega} \in Q^\omega$ where $q_n = q_3$ for infinitely many $n \in \omega$. Figure \ref{fig:prerunExample} depicts the $F$-structure $\rho_F$ of an accepting pre-run $(R, \rho = \langle \rho_F, \rho_X, \rho_Q \rangle, r_1)$, with $\rho_X$ and $\rho_Q$ given by:
    \begin{align*}
        \rho_X:& \qquad r_1 \mapsto x_1, \quad r_2 \mapsto x_2, \quad \{ r_3, r_4, \dotsc \} \mapsto x_3, \\
        \rho_Q:& \qquad r_1 \mapsto q_1, \quad r_2 \mapsto q_2, \quad \{ r_3, r_4, r_6, r_7, \dotsc, r_{3n}, r_{3n+1}, \dotsc \} \mapsto q_3, \quad \{ r_5, r_8, \dotsc, r_{3n+2}, \dotsc \} \mapsto q_4.
    \end{align*}
    Note that $(R, \rho, r_1)$ is not (isomorphic to) a run, because $r_3$ and $r_6$ are behaviourally equivalent. Figure \ref{fig:runExample} shows the $F$-structure of another accepting pre-run $(R', \rho' = \langle \rho_F', \rho_X', \rho_Q' \rangle, r_1)$, with $\rho'_X$ and $\rho'_Q$ given by:
    \begin{align*}
        \rho_X':& \qquad r_1 \mapsto x_1, \quad r_2 \mapsto x_2, \quad \{ r_3, r_4, r_5 \} \mapsto x_3, \\
        \rho_Q':& \qquad r_1 \mapsto q_1, \quad r_2 \mapsto q_2, \quad \{r_3, r_4\} \mapsto q_3, \quad  r_5 \mapsto q_4.
    \end{align*}
    One can check that $(R', \rho', r_1)$ is isomorphic to a run, since no two elements of $R'$ are behaviourally equivalent. Moreover, by merging $r_3$ and $r_4$ (i.e., setting $\rho_F'(r_3) = (r_5)$ and dropping $r_4$), we obtain another accepting run. Therefore $\mathcal A$ is not unambiguous on $(X,\xi,x_1)$.
\end{example}

Below we state two basic properties of (pre-)runs.

\begin{lemma}
\label{lem:preRunImpliesAccepted}
    Let $\mathcal A$ be an $F$-automaton and $(X, \xi, x_I)$ be a pointed $F$-coalgebra. If $(R, \rho, r_I)$ is an accepting pre-run of $\mathcal A$ on $(X,\xi,x_I)$, then its image in the final $(F \times \Delta_X \times \Delta_Q)$-coalgebra is an accepting run. Hence $\mathcal A$ accepts $(X, \xi, x_I)$ if and only if $\mathcal A$ has an accepting pre-run on $(X, \xi, x_I)$.
\end{lemma}

\begin{proposition}
\label{prop:acceptanceInvariantUnderMorphisms}
    Let $\mathcal A$ be an $F$-automaton and $f\colon (X, \xi, x_I) \to (X', \xi', x_I')$ be a pointed $F$-coalgebra morphism. Then $\mathcal A$ accepts $(X, \xi, x_I)$ if and only if $\mathcal A$ accepts $(X', \xi', x_I')$.
\end{proposition}

We conclude this section with a strengthening of \cite[Theorem~4.4]{KupkeVenema2008CoalgebraicAutomata}: every automaton with $\omega$-regular acceptance can be transformed into an equivalent parity automaton in an unambiguity-preserving way.

\begin{proposition}
\label{prop:omegaRegularToParity}
    Every $F$-automaton $\mathcal A$ with $\omega$-regular acceptance can be transformed into a parity $F$-automaton $\mathcal A'$ accepting the same coalgebras. Moreover, for every pointed $F$-coalgebra $(X, \xi, x_I)$, if $\mathcal A$ is unambiguous on $(X, \xi, x_I)$, then $\mathcal A'$ is unambiguous on $(X, \xi, x_I)$.
\end{proposition}

\section{From Automata to Algebras}\label{sec:automatonToAlgebra}
In this section, we describe the first central construction of the paper: the automaton algebra. Given an $F$-automaton with prefix-agnostic acceptance, we construct a finite coherent $(F+G)$-algebra with a recognising set, which recognises the language consisting of the thin behaviours that are accepted by the $F$-automaton. Our construction is inspired by the construction of the \emph{thin algebra of an automaton} for binary trees~\cite[Section~6.2.1]{Skrzypczak2016}.

\subsection{The Automaton Algebra}
Given an $F$-automaton $\mathcal A$, the elements of the automaton algebra $\autalg_{\mathcal A}$ will be sets of automaton states. 
The algebra structure of $\autalg_{\mathcal A}$ is defined in order to obtain the following property:
if $z \in \thin Z$, then $\cev_{\autalg_\mathcal A}(z)$ is the set of those states $q$ such that $\mathcal A$ has a run of $(\thin Z, \thin \zeta, z)$, starting at $q$ (recall $\cev_{\autalg_ {\mathcal A}}: (\thin Z, \thin \beta) \to \autalg_{\mathcal A}$).

\begin{definition}
    Let $\mathcal A = (Q, \delta, Q_I, \Acc)$ be an $F$-automaton with prefix-agnostic $\Acc$. Define the \emph{automaton algebra} $\autalg_{\mathcal A} \coloneqq (C, [\gamma_0,\gamma_1], U)$ of $\mathcal A$ as follows.
    \begin{itemize}
        \item $C \coloneqq \Pow(Q)$;
        \item for all $\bar c \in FC$: $\gamma_0(\bar c) \coloneqq \{ q \in Q \mid \exists \bar q \in FQ (\bar q \inlift \bar c \land \bar q \in \delta(q)) \}$;
        \item for all $(\bar c'_n)_{n \in \omega} \in (F'C)^\omega$:
        \begin{align*}
            \gamma_1((\bar c'_n)_{n \in \omega}) &\coloneqq \{ q_0 \in Q \mid \exists (q_n)_{n \in \omega} \in \Acc, (\bar q'_n)_{n \in \omega} \in (F'Q)^\omega:
            \forall n \in \omega (\bar q'_n \inlift \bar c'_n), \\
            & \hspace{234px} \forall n \in \omega (\trig_Q(\bar q'_n, q_{n+1}) \in \delta(q_n)) \};
        \end{align*}
        \item $U \coloneqq \{ c \in C \mid c \cap Q_I \neq \emptyset \}$.
    \end{itemize}
\end{definition}

For simplicity, consider a polynomial functor $F$. In the definition of $\gamma_0$, $\gamma_0(\bar c)$ consists of those states $q$, for which there exists a transition $\bar q \in \delta(q)$ such that each component in the tuple $\bar q$ is an element of the corresponding component of $\bar c$. This corresponds to the fact that a thin behaviour $z$ is accepted by $\mathcal A$, starting at $q$, precisely when there exists a transition $\bar q \in \delta(q)$ such that, for all $i$, $\mathcal A$  accepts the $i$-th successor of $z$, starting at the $i$-th component of $\bar q$. Here it is essential to assume $\Acc$ is prefix-agnostic, so that for all $p \in \Base_F(\bar q)$ and $x \in Q^\omega$, we have $qpx \in \Acc$ if and only if $px \in \Acc$.

Similarly, $\gamma_1((\bar c'_n)_{n \in \omega})$ consists of states $q$ such that we can choose a context $\bar q'_n$ of states for every context $\bar c'_n$, and a sequence of states $(q_n)_{n \geq 1} \in \Acc$ to fill the consecutive holes in these contexts. Again, we use the prefix-agnostic assumption, so that for all $p \in \Base_{F'}(\bar q'_n)$ and $x \in Q^\omega$, we have $q_0q_1\dotsc q_n p x \in \Acc$ if and only if $px \in \Acc$.

For the recognising set $U$, we take those sets of states $c$ that contain at least one accepting state, so that $\cev_{\autalg_{\mathcal A}}^{-1}(U)$ contains the thin behaviours accepted by $\mathcal A$.

Theorem~\ref{thm:automatonAlgebraLanguage} below connects acceptance of an automaton $\mathcal A$ with the language 
$L(\autalg_{\mathcal A})$ of its automaton algebra. It uses the key property that the automaton algebra is coherent.

\begin{lemma}
\label{lem:autAlgebraCoherent}
    For all $F$-automata $\mathcal A$ with prefix-agnostic acceptance, the automaton algebra $\autalg_{\mathcal A}$ is coherent.
\end{lemma}

\begin{theorem}
\label{thm:automatonAlgebraLanguage}
    Let $\mathcal A$ be an $F$-automaton with prefix-agnostic acceptance and let $(X, \xi, x_I)$ be a thin pointed $F$-coalgebra. Then $\mathcal A$ accepts $(X, \xi, x_I)$ if and only if 
    $\tbeh_{(X,\xi)}(x_I) \in L(\autalg_{\mathcal A})$.
\end{theorem}
\begin{proof*}{Proof (Sketch)}
  Let $\mathcal A = (Q, \delta, Q_I, \Acc)$ and $\autalg_{\mathcal{A}} = (C, \gamma = [\gamma_0,\gamma_1], U)$. Define:
    \begin{align*}
        &f : \thin Z \to C, \\
        &z \mapsto \{ q \in Q \mid \text{there exists a run of $\mathcal A$ on $(\thin Z,\thin \zeta, z)$, starting from $q$}\}.
    \end{align*}
    One can show that $f: (\thin Z,\thin \beta) \to (C, \gamma)$ is an $(F+G)$-algebra morphism. Since $(\thin Z, \thin \beta)$ is an initial coherent $(F+G)$-algebra, this implies $f = \cev_{(C, \gamma)}$. Now $\mathcal A$ accepts $(X, \xi, x_I)$ if and only if $\mathcal A$ accepts $(\thin Z, \thin \zeta, \tbeh_{(X, \xi)}(x_I))$ (by Proposition~\ref{prop:acceptanceInvariantUnderMorphisms}) if and only if $(f \circ \tbeh_{(X, \xi)})(x) \cap Q_I \neq \emptyset$ if and only if $(\cev_{(C,\gamma)} \circ \tbeh_{(X,\xi)})(x) \in U$ if and only if $\tbeh_{(X,\xi)}(x_I) \in L(\autalg_\mathcal A)$.
\end{proof*}

\begin{example}
\label{ex:automatonAlgebra}
    Let $\mathcal A$ be the automaton from Example \ref{ex:automatonNonRegularLanguage}. Its automaton algebra $\autalg_{\mathcal A} = (C, [\gamma_0,\gamma_1], U)$ has a carrier $C = \{ \emptyset, \{q_a\}, \{q_b\}, \{q_a,q_b\} \}$. The $F$-operation is given by $\gamma_0(\sigma, c) = \{ q_a, q_b \}$, if $q_\sigma \in c$, and $\gamma_0(\sigma, c) = \emptyset$, otherwise (for all $\sigma \in \Sigma$ and $c \in C$). For the $G$-operation, for every $(\sigma_n)_{n \in \omega} \in GC \,\cong\, \Sigma^\omega$, we have that $\gamma_1((\sigma_n)_{n \in \omega})$ equals $\{ q_a, q_b \}$, if $(\sigma_n)_{n\in\omega} \in L$, and $\emptyset$, otherwise. For the recognising set, we have $U = \{ \{q_a \}, \{q_b\}, \{q_a, q_b \} \}$. If we take $Z = \Sigma^\omega$ (the final coalgebra of streams over $\Sigma$), we get $L(\autalg_{\mathcal A}) = L \subseteq \thin Z = Z$.
\end{example}

\subsection{Rational Algebras}
Example \ref{ex:automatonAlgebra} showed that there exist finite coherent $(F+G)$-algebras whose language cannot be characterised by parity $F$-automata. A finite coherent $(F+G)$-algebra $(C, \gamma)$ partitions $\thin Z$ into finitely many classes $\{ \cev_{(C,\gamma)}^{-1}(c) \mid c \in C \}$. In order to retain the connection to parity $F$-automata, in Definition \ref{def:rational} we equip $(C,\gamma, U)$ with additional structure so that it also partitions into finitely many classes the set $(F'\thin Z)^+$ of finite sequences of contexts over $\thin Z$. Intuitively, a sequence of $n$ contexts is viewed as the ``nested context'' obtained by plugging the sequence together, so that the hole is at depth $n$ (whereas in our usual contexts the hole is at depth $1$). The partition of $(F'\thin Z)^+$ is to satisfy the following property: if $(\widetilde c_n)_{n\in\omega}, (\widetilde d_n)_{n\in\omega} \in ((F'\thin Z)^+)^\omega$ and for all $n \in \omega$, $\widetilde c_n$ and $\widetilde d_n$ are in the same class, then $\gamma_1(\widetilde c_0 \widetilde c_1 \dotsc) \in L(C,\gamma,U)$ if and only if $\gamma_1(\widetilde d_0 \widetilde d_1 \dotsc) \in L(C,\gamma,U)$. Note that in Example \ref{ex:automatonAlgebra} it is impossible to find a finite partition of $(F'\thin Z)^+ = \Sigma^+$ with this property. In order to guarantee the property, we define the following subclass of finite coherent $(F+G)$-algebras.

\begin{definition}
\label{def:rational}
    Let $(C, \gamma = [\gamma_0, \gamma_1])$, be a finite coherent $(F+G)$-algebra and $\Sigma \coloneqq F'C$. We call $(C, \gamma)$ \emph{rational} if there exists a finite $\omega$-semigroup $(\widetilde C, \Im(\gamma_1))$ and a map $\gamma_2: \Sigma^+ \to \widetilde C$ such that $(\gamma_2, \gamma_1): (\Sigma^+, \Sigma^\omega) \to (\widetilde C, \Im(\gamma_1))$ is an $\omega$-semigroup homomorphism.
\end{definition}

In the above definition, the map $\gamma_2: \Sigma^+ \to \widetilde C$ partitions the set of finite sequences of contexts (i.e., the nested contexts) into finitely many classes $\widetilde C$.

Note that for a functor $FX = \Sigma_0 + \Sigma_2 \times X \times X$, where $\Sigma_0$ and $\Sigma_2$ are alphabets, rational $(F+G)$-algebras essentially coincide with thin algebras~\cite{Skrzypczak2016}. Thin algebras contain two sorts: a sort for trees (in rational algebras, this is the domain $C$) and a sort for contexts (in rational algebras, this is the set $\widetilde C$). Hence rational $(F+G)$-algebras can be seen as a natural generalisation of thin algebras to analytic functors.

We will see in Section \ref{sec:mainResults} that languages of rational $(F+G)$-algebras can be characterised by parity $F$-automata. For now, we only show that parity $F$-automata give rise to rational $(F+G)$-algebras.

\begin{proposition}
\label{prop:automatonAlgebraRational}
    For all $F$-automata $\mathcal A$ with parity acceptance, the automaton algebra $\autalg_{\mathcal A}$ is rational.
\end{proposition}
\begin{proof*}{Proof (Sketch)}
    The construction generalises \cite[Section~6.2.1]{Skrzypczak2016}.
    Let $\mathcal A = (Q, \delta, Q_I, \Omega)$ be a parity $F$-automaton, and $\autalg_{\mathcal A} = (C, \gamma = [\gamma_0, \gamma_1], U)$. Define a two-sorted algebra $(\widehat C, C)$ by:
    \begin{align*}
        \widehat C &\coloneqq \mathcal P(Q \times Q \times \Im(\Omega)), \\
        \widehat c_1 \cdot \widehat c_2 &\coloneqq \{ (q, q_2, max\{m_1, m_2\}) \mid \exists q_1 \in Q: (q, q_1, m_1) \in \widehat c_1 \land (q_1, q_2, m_2) \in \widehat c_2 \}, \\
        \widehat c \times c &\coloneqq \{ q \mid \exists q_1 \in c, m \in \omega: (q, q_1, m) \in \widehat c \, \}, \\
        \Pi((\widehat c_n)_{n\in\omega}) &\coloneqq \{ q_0 \mid \exists (q_n)_{n\in\omega} \in Q^\omega, (m_n)_{n\in\omega} \in \mathbb{N}^\omega: 
        \forall n \in \omega ((q_n, q_{n+1}, m_n) \in \widehat c_n) \: \land 
        \limsup_{n\in\omega} m_n \text{ is even} \}.
    \end{align*}
    We define the map $\gamma_2: (F'C)^+ \to \widehat C$ by specifying its restriction to the set of generators $F'C$ of the freely generated semigroup $(F'C)^+$. For $\bar c' \in F'C$, we set:
    \begin{equation*}
        \gamma_2(\bar c') \coloneqq \{ (q, q_1, max\{\Omega(q), \Omega(q_1)\}) \mid \exists \bar q' \in F'Q (\bar q' \inlift \bar c' \land \trig_Q(\bar q', q_1) \in \delta(q) \}.
    \end{equation*}
    One can show that $(\Im(\gamma_2), \Im(\gamma_1))$ is an $\omega$-semigroup and $(\gamma_2, \gamma_1)$ is a homomorphism.
\end{proof*}

\section{From Algebras to Automata}\label{sec:algebraToAutomaton}
In this section, we show how to construct from a finite coherent algebra its \emph{algebraic automaton}. 
The context decomposition transformation from Definition~\ref{def:decomp} is instrumental in defining the transition structure of this automaton. The key result here is that the algebraic automaton is unambiguous on thin coalgebras. We proceed as follows: we introduce the algebraic automaton, develop the key technical notion of \emph{marking} and use it to show that, when restricting to thin $F$-coalgebras, the algebraic automaton is unambiguous and accepts the same language as the starting algebra.

\subsection{The Algebraic Automaton}
Given a finite coherent $(F+G)$-algebra $(C, [\gamma_0,\gamma_1], U)$ with a recognising set, we aim to construct an equivalent unambiguous automaton.
We draw inspiration from the construction in \cite[Section~7.2.1]{Skrzypczak2016} for binary trees. The idea is that each state $q$ in the algebraic automaton encodes an element $c \in C$ in such a way that the algebraic automaton accepts, starting at state $q$, those pointed coalgebras $(\thin Z, \thin \zeta,z)$ for which $\cev_{(C,\gamma)}(z) = c$. A run of the algebraic automaton labels behaviours $z \in \thin Z$ with algebra elements $c \in C$. The transitions of the automaton are to ensure that if $z$ is labelled with $c \in C$ and $\thin \zeta(z) \in F\thin Z$ is labelled with $\bar c \in FC$, then $\gamma_0(\bar c) = c$.
The acceptance condition is to ensure that for every infinite path $(z, z_1, z_2 \dotsc)$, 
if $z$ is labelled with $c$, and for all $n \geq 1$, the context in  $F'\thin Z$ consisting of the siblings of $z_n$ is labelled with $\bar c'_n \in F'C$, then $\gamma_1((\bar c'_n)_{n \geq 1}) = c$.
In order to realise the latter requirement, a state $q$ of the automaton must encode both a label $c \in C$ for a behaviour and a context of labels $\bar c' \in F'C$ for the context of siblings of that behaviour.
Since the root of a pointed $F$-coalgebra does not have any siblings, we need additional states that only encode a label in $C$ -- these states occur only in the root of the run.

Recall that the notion ``context of siblings'' can be expressed formally using the context decomposition operator $\decomp: F \Rightarrow F(F' \times \Id)$ from Definition \ref{def:decomp}.

\begin{definition}
\label{def:algebraicAutomaton}
    Let $\mathsf{C} = (C, [\gamma_0, \gamma_1], U)$ be a finite coherent $(F+G)$-algebra with a recognising set. Define the \emph{algebraic automaton} $\algaut_{\mathsf{C}} \coloneqq (Q, \delta, Q_I, \Acc)$ as follows:
    \begin{itemize}
        \item $Q \coloneqq C + F'C \times C$;
        \item $Q_I \coloneqq \inj{1}[U]$;
        \item $\delta(\inj{1}(c)) \coloneqq \{ (F\inj{2} \circ \decomp_C)(\bar c)) \mid \gamma_0(\bar c) = c \}$, for $c \in C$,
        \quad
        $\delta(\inj{2}(\bar c', c)) \coloneqq \delta(\inj{1}(c))$, for $(\bar c', c) \in F'C \times C$;
        \item $\Acc \coloneqq \{ \inj{1}(c_0) \cdot (\inj{2}(\bar c'_n, c_n))_{n > 0} \mid \forall m (c_m = \gamma_1((\bar c'_n)_{n > m})) \}$.
    \end{itemize}
\end{definition}

In the above definition of $\delta$, transitions from an automaton state labelled with $c \in C$ cover all possible decompositions of all $\bar c$ such that $\gamma_0(\bar c) = c$. The algebraic automaton is defined such that it accepts the same thin behaviours as the corresponding coherent algebra. Furthermore, it has precisely one run on each thin coalgebra, thus it is unambiguous. The rest of the section is dedicated to proving these statements. 

\subsection{Markings}
In order to relate pre-runs of the algebraic automaton with the corresponding coherent algebra, we introduce the notion of \emph{marking}, which generalises consistent markings on binary trees~\cite[Section~7.1]{Skrzypczak2016}.

\begin{definition}
    Let $(C, \gamma = [\gamma_0, \gamma_1])$ be a coherent $(F+G)$-algebra and let $(X, \xi)$ be an $F$-coalgebra. A~\emph{marking} of $(X,\xi)$ with $(C,\gamma)$ is a map $\mu: X \to C$ satisfying:
    \begin{enumerate}
        \item[(i)] $\mu: (X, \xi) \to (C, \gamma_0)$ is an $F$-coalgebra-to-algebra morphism;
        \item[(ii)] for all $(x_n)_{n\in\omega} \in X^\omega$, $(\bar x_n')_{n > 0} \in GX$ with $\forall n(\trig_X(\bar x'_{n+1}, x_{n+1}) = \xi(x_n))$:  $\gamma_1(G\mu((\bar x_n')_{n>0})) = \mu(x_0)$.
    \end{enumerate}
\end{definition}

Roughly speaking, property (i) of markings is the algebraic counterpart to property (ii) of pre-runs of the algebraic automaton, while property (ii) of markings is the algebraic counterpart to property (iii) of pre-runs. So, intuitively, pre-runs of the algebraic automaton compute a marking. The precise connection between markings and pre-runs of the algebraic automaton is given in the following statement.

\begin{proposition}
\label{prop:markingsAndPreruns}
    Let $\mathsf C = (C, \gamma = [\gamma_0,\gamma_1], U)$ be a finite coherent $(F+G)$-algebra with a recognising set, and let $(X, \xi, x_I)$ be a pointed $F$-coalgebra.
    \begin{enumerate}
        \item[(i)] If $(R, \rho = \langle \rho_F, \rho_X, \rho_Q \rangle, r_I)$ is a pre-run of $\algaut_{\mathsf{C}}$ on $(X, \xi, x_I)$, then $[\id, \prj{2}] \circ \rho_Q: R \to C$ is a marking of $(R, \rho_F)$ with $(C, \gamma)$ (see Figure \ref{fig:runToMarkings}).
        \item[(ii)] If $\mu: X \to C$ is a marking of $(X,\xi)$ with $(C,\gamma)$, then there exists a pre-run $(R,\rho = \langle \rho_F, \rho_X, \rho_Q \rangle, r_I)$ of $\algaut_{\mathsf{C}}$ on $(X,\xi, x_I)$ with $\rho_Q(r_I) \in \inj{1}[C]$ and $[\id,\prj{2}] \circ \rho_Q = \mu \circ \rho_X$ (see Figure \ref{fig:markingsToRun}).
    \end{enumerate}
    \begin{figure}[ht]
        \centering
        \begin{subfigure}[b]{0.5\textwidth}
            \centering
            \begin{tikzcd}
                R & {C + (F'C \times C)} & C \\
                FR && FC
                \arrow["{{\rho_Q}}", from=1-1, to=1-2]
                \arrow["{\rho_F}", from=1-1, to=2-1]
                \arrow["{{[\id, \prj{2}]}}", from=1-2, to=1-3]
                \arrow["{F([\id,\prj2] \circ \rho_Q)}"', from=2-1, to=2-3]
                \arrow["{\gamma_0}"', from=2-3, to=1-3]
            \end{tikzcd}
            \caption{Property (i)}
            \label{fig:runToMarkings}
        \end{subfigure}%
        \begin{subfigure}[b]{0.5\textwidth}
            \centering
            \begin{tikzcd}
                FR & R & {C+(F'C\times C)} \\
                FX & X & C
                \arrow["{F\rho_X}"', from=1-1, to=2-1]
                \arrow["{\rho_F}"', from=1-2, to=1-1]
                \arrow["{\rho_Q}", from=1-2, to=1-3]
                \arrow["{\rho_X}", from=1-2, to=2-2]
                \arrow["{[\id,\prj2]}", from=1-3, to=2-3]
                \arrow["\xi", from=2-2, to=2-1]
                \arrow["\mu"', from=2-2, to=2-3]
            \end{tikzcd}
            \caption{Property (ii)}
            \label{fig:markingsToRun}
        \end{subfigure}
        \caption{Diagrams for Proposition \ref{prop:markingsAndPreruns}.}
        \label{fig:markingsAndRuns}
    \end{figure}
\end{proposition}
\begin{proof*}{Proof (Sketch)}
    \textbf{(i).} It can be verified that $[\id, \prj2] \circ \rho_Q: R \to C$ satisfies the properties of markings, using Lemma \ref{lem:decompProperties}. The proof of property (i) of markings uses property (ii) of the pre-run $(R, \rho, r_I)$, while for property (ii) of markings we use property (iii) of pre-runs.

    \textbf{(ii).} We define a pointed $(F \times \Delta_X \times \Delta_Q)$-coalgebra $\mathsf{R} \coloneqq (R, \langle \rho_F,\rho_X,\rho_Q \rangle, r_I)$ with $R \coloneqq X \times Q$, $\rho_X \coloneqq \prj{1}$ and $\rho_Q \coloneqq \prj{2}$. The marking $\mu$ is used to define $r_I \coloneqq (x_I, \inj{1} \circ \mu(x_I))$ and:
    \begin{align*}
        \rho_F &\coloneqq X \times Q \xrightarrow{\prj{1}} X \xrightarrow{\xi} FX \xrightarrow{\decomp_X} F(F'X \times X) \xrightarrow{F(F'\mu \times \langle \id, \mu \rangle)} F(F'C \times (X \times C)) \\
        &\hspace{154px} \xrightarrow{\cong} F(X \times (F'C \times C)) \xrightarrow{F(\id \times \inj{2})} F(X \times Q).
    \end{align*}
    It can be verified that the reachable subcoalgebra of $\mathsf{R}$ is a pre-run, using Lemma \ref{lem:decompProperties} and properties of the marking $\mu$.
\end{proof*}

The benefit of working with markings instead of (pre-)runs is that markings are defined solely in terms of the algebra, as opposed to in terms of the algebraic automaton. We will see in Lemma \ref{lem:F+GcoalgebraOnX} that by equipping $(X,\xi)$ with a suitable $(F+G)$-coalgebra structure, markings turn into $(F+G)$-coalgebra-to-algebra morphisms. This will allow us to find existence and uniqueness properties of markings that follow from the recursive structure of thin behaviours. Consequently, Proposition \ref{prop:markingsAndPreruns} will allow us to draw conclusions about pre-runs of the algebraic automaton.

\paragraph{Properties of Markings} The first property of markings is that every thin coalgebra can be marked. Concretely, for all thin coalgebras $(X,\xi)$ and all coherent algebras $(C,\gamma)$, we show that $\mu \coloneqq \cev_{(C,\gamma)} \circ \tbeh_{(X,\xi)}$ is a marking of $(X,\xi)$ with $(C,\gamma)$. Our strategy is to show that $\cev_{(C,\gamma)}$ is a marking and that markings are preserved under precomposition with $F$-coalgebra morphisms.

\begin{lemma}
\label{lem:cevIsMarking}
    If $(C,\gamma)$ is a coherent $(F+G)$-algebra, then $\cev_{(C,\gamma)}: (\thin Z, \thin \beta) \to (C,\gamma)$ is a marking of $(\thin Z, \thin \zeta)$ with $(C, \gamma)$.
\end{lemma}
\begin{proof} \qedoverwrite
    To see that $\cev_{(C, \gamma)}$ satisfies condition (i) of markings, i.e., $\cev_{(C, \gamma)}: (\thin Z, \thin \zeta) \to (C, \gamma_0)$ is an $F$-coalgebra-to-algebra morphism, consider the diagram to the right. We have:

    \noindent\begin{minipage}{0.75\textwidth}
        \begin{equation*}
            \cev_{(C, \gamma)} \circ \thin \beta_0 = \gamma_0 \circ F\cev_{(C, \gamma)} = \gamma_0 \circ F\cev_{(C, \gamma)} \circ \thin \zeta \circ \thin \beta_0,
        \end{equation*}
        where the first equality uses that $\cev_{(C, \gamma)}$ is an $(F+G)$-algebra morphism and the second equality uses Equation \eqref{eq:betaZetaProp1}. Now since $\thin \beta_0$ is epic, we conclude $\cev_{(C, \gamma)} = \gamma_0 \circ F\cev_{(C, \gamma)} \circ \thin \zeta$, i.e., $\cev_{(C, \gamma)}$ is an $F$-coalgebra-to-algebra morphism.
    \end{minipage}%
    \begin{minipage}{0.25\textwidth}
        \vspace{-20px}
        \begin{equation*}
            \begin{tikzcd}
                {F(\thin Z)} & {F(C)} \\
                {\thin Z} & C
                \arrow["{F\cev}", from=1-1, to=1-2]
                \arrow["{\thin \beta_0}", from=1-1, to=2-1]
                \arrow["{\gamma_0}", from=1-2, to=2-2]
                \arrow["{\thin \zeta}", curve={height=-6pt}, from=2-1, to=1-1]
                \arrow["\cev"', from=2-1, to=2-2]
            \end{tikzcd}
        \end{equation*}
    \end{minipage}

    To see that $\cev_{(C, \gamma)}$ satisfies condition (ii) of markings, let $(z_n)_{n \in \omega} \in (\thin Z)^\omega$ and $(\bar z_n')_{n > 0} \in G\thin Z$ satisfy $\trig_{\thin Z}(\bar z_{n+1}', z_{n+1}) = \thin \zeta(z_n)$, for all $n \in \omega$. It follows from Equation \eqref{eq:betaZetaProp2} that $\thin \beta_1((\bar z_n')_{n > 0}) = z_0$. Hence:
    \begin{equation*}
        \cev_{(C, \gamma)}(z_0) = \cev_{(C, \gamma)}(\thin \beta_1((\bar z_n')_{n > 0})) = (\gamma_1 \circ G\cev_{(C,\gamma)}) ((\bar z_n')_{n > 0}). \qedhere
    \end{equation*}
\end{proof}

\begin{lemma}
\label{lem:markingsPreservedByPrecomp}
    If $\mu: (X, \xi) \to (C, \gamma)$ is a marking and $f: (Y,\upsilon) \to (X,\xi)$ is an $F$-coalgebra morphism, then $\mu \circ f$ is a marking.
\end{lemma}

\begin{proposition}[Existence of Markings]
\label{prop:existenceMarkings}
    For every thin $F$-coalgebra $(X, \xi)$ and every coherent $(F+G)$-algebra $(C, \gamma)$, there exists a marking $\mu$ of $(X, \xi)$ with $(C, \gamma)$ given by $\mu = \cev_{(C,\gamma)} \circ \tbeh_{(X,\xi)}$.
    \begin{equation*}
        \begin{tikzcd}
            X & {\thin Z} & C \\
            FX & {F\thin Z} & FC
            \arrow["{\tbeh_{(X,\xi)}}", from=1-1, to=1-2]
            \arrow["\xi"', from=1-1, to=2-1]
            \arrow["{\cev_{(C,\gamma)}}", from=1-2, to=1-3]
            \arrow["{F\tbeh_{(X,\xi)}}"', from=2-1, to=2-2]
            \arrow["{\thin \zeta}"', from=2-2, to=1-2]
            \arrow["{F\cev_{(C,\gamma)}}"', from=2-2, to=2-3]
            \arrow["{\gamma_0}"', from=2-3, to=1-3]
        \end{tikzcd}
    \end{equation*}
\end{proposition}
\begin{proof}
    By Lemma \ref{lem:cevIsMarking}, $\cev_{(C,\gamma)}: \thin Z \to C$ is a marking of $(\thin Z, \thin \zeta)$ with $(C, \gamma)$. By Lemma \ref{lem:markingsPreservedByPrecomp}, $\cev_{(C, \gamma)} \circ \tbeh_{(X, \xi)}$ is a marking of $(X,\xi)$ with $(C, \gamma)$.
\end{proof}

The second central property of markings is uniqueness: there do not exist two distinct markings of a given thin coalgebra with a given coherent algebra (Proposition \ref{prop:uniquenessOfMarkings}). The key insight behind the proof is that every thin coalgebra $(X, \xi)$ can be transformed into a recursive $(F+G)$-coalgebra such that markings of $(X,\xi)$ become $(F+G)$-coalgebra-to-algebra morphisms (Lemma \ref{lem:F+GcoalgebraOnX}). Uniqueness of markings will then follow from the fact that coalgebra-to-algebra morphisms with a recursive domain coincide.

The recursive $(F+G)$-coalgebra structure on $X$ is inherited from a canonical recursive $(F+G)$-coalgebra structure $\eta$ on $\thin Z$. Intuitively, $\eta$ decomposes a normal term $z \in \thin Z$ into its normal subterms.

\begin{definition}
\label{def:etaMap}
    Let $\iota: \thin Z \to A$ be the map sending each thin behaviour to its unique normal representative. Define an $(F+G)$-coalgebra structure $\eta$ on $\thin Z$ by $\eta \coloneqq (F+G)\interpr- \circ \alpha^{-1} \circ \iota$.
    \begin{equation*}
        \begin{tikzcd}
            A & {\thin Z} \\
            {(F+G)A} & {(F+G)\thin Z}
            \arrow["{\interpr-}"', curve={height=12pt}, from=1-1, to=1-2]
            \arrow["{\alpha^{-1}}"', from=1-1, to=2-1]
            \arrow["\iota"', curve={height=12pt}, from=1-2, to=1-1]
            \arrow["\eta", from=1-2, to=2-2]
            \arrow["{(F+G)\interpr-}"', curve={height=14pt}, from=2-1, to=2-2]
        \end{tikzcd}
    \end{equation*}
\end{definition}

Next, we show that we can define an $(F+G)$-coalgebra structure on any thin coalgebra that turns markings into $(F+G)$-coalgebra-to-algebra morphisms.

\begin{lemma}
\label{lem:F+GcoalgebraOnX}
    Let $(X, \xi)$ be a thin $F$-coalgebra and $(C, \gamma)$ be a coherent $(F+G)$-algebra. There exists an $(F+G)$-coalgebra structure $\upsilon$ on $X$ such that:
    \begin{enumerate}
        \item[(i)] $\tbeh \coloneqq \tbeh_{(X,\xi)}$ is an $(F+G)$-coalgebra morphism from $(X, \upsilon)$ to $(\thin Z, \eta)$, and
        \item[(ii)] every marking $\mu: X \to C$ is an $(F+G)$-coalgebra-to-algebra morphism $(X,\upsilon) \to (C,\gamma)$.
    \end{enumerate}
    \[\begin{tikzcd}
        & FX & {F\thin Z} \\
        C & X & {\thin Z} \\
        {(F+G)C} & {(F+G)X} & {(F+G)\thin Z}
        \arrow["F\tbeh", from=1-2, to=1-3]
        \arrow["\xi"', from=2-2, to=1-2]
        \arrow["\mu"', from=2-2, to=2-1]
        \arrow["\tbeh", from=2-2, to=2-3]
        \arrow["\upsilon", dashed, from=2-2, to=3-2]
        \arrow["{\thin \zeta}"', from=2-3, to=1-3]
        \arrow["{\eta}", from=2-3, to=3-3]
        \arrow["\gamma", from=3-1, to=2-1]
        \arrow["{(F+G)\mu}", from=3-2, to=3-1]
        \arrow["{(F+G)\tbeh}"', from=3-2, to=3-3]
    \end{tikzcd}\]
\end{lemma}

\begin{proposition}[Uniqueness of Markings]
\label{prop:uniquenessOfMarkings}
    For every thin $F$-coalgebra $(X, \xi)$ and every coherent $(F+G)$-algebra $(C, \gamma)$, there is at most one marking of $(X, \xi)$ with $(C, \gamma)$.
\end{proposition}
\begin{proof}
    Let $\mu_1$ and $\mu_2$ be two markings of $(X,\xi)$ with $(C, \gamma)$.
    By appealing to Lemma \ref{lem:F+GcoalgebraOnX}, we obtain a coalgebra structure $\upsilon: X \to (F+G)X$. Consider the diagram:
    \[\begin{tikzcd}
        C && X && {\thin Z} & A \\
        {(F+G)C} && {(F+G)X} && {(F+G)\thin Z} & {(F+G)A}
        \arrow["{\mu_1}"', shift right, from=1-3, to=1-1]
        \arrow["{\mu_2}", shift left, from=1-3, to=1-1]
        \arrow["{\tbeh_{(X,\xi)}}", from=1-3, to=1-5]
        \arrow["\upsilon", from=1-3, to=2-3]
        \arrow["\iota", from=1-5, to=1-6]
        \arrow["\eta", from=1-5, to=2-5]
        \arrow["{\alpha^{-1}}", from=1-6, to=2-6]
        \arrow["\gamma", from=2-1, to=1-1]
        \arrow["{(F+G)\mu_1}"', shift right, from=2-3, to=2-1]
        \arrow["{(F+G)\mu_2}", shift left, from=2-3, to=2-1]
        \arrow["{(F+G)\tbeh_{(X,\xi)}}"', from=2-3, to=2-5]
        \arrow["{(F+G)\iota}"', from=2-5, to=2-6]
    \end{tikzcd}\]
    We know $\mu_1$, $\mu_2$ are $(F+G)$-coalgebra-to-algebra morphisms and that $\tbeh_{(X,\xi)}$, $\iota$ are $(F+G)$-coalgebra morphisms. Observe that, since $(A,\alpha)$ is an initial $(F+G)$-algebra, the coalgebra $(A,\alpha^{-1})$ is recursive~\cite[Corollary~8.2]{AdamekEtAlWellFoundedRecursiveCoalgebras}. Moreover, any coalgebra mapping into a recursive coalgebra is also recursive~\cite[Corollary~8.2]{AdamekEtAlWellFoundedRecursiveCoalgebras}, hence $(X, \upsilon)$ is recursive. Now $\mu_1$ and $\mu_2$ are two coalgebra-to-algebra morphisms with a recursive coalgebra as their domain, therefore $\mu_1 = \mu_2$.
\end{proof}

In the proof of Proposition \ref{prop:uniquenessOfMarkings}, note the instrumental role of the inductive structure of thin behaviours. It is what allowed us to obtain a recursive $(F+G)$-coalgebra structure on $X$.

\subsection{Acceptance and Unambiguity of the Algebraic Automaton}

\begin{theorem}
\label{thm:algebraicAutomatonLanguage}
    Let $\mathsf C = (C, \gamma = [\gamma_0, \gamma_1], U)$ be a coherent $(F+G)$-algebra with a recognising set. For every thin pointed $F$-coalgebra $(X, \xi, x_I)$, the behaviour of $x_I$ is in the language of $\mathsf C$ if and only if the algebraic automaton $\algaut_{\mathsf{C}}$ accepts $(X, \xi, x_I)$.
\end{theorem}
\begin{proof}
    Suppose $(\cev_{(C, \gamma)} \circ \tbeh_{(X,\xi)})(x_I) \in U$. By Proposition \ref{prop:existenceMarkings}, $\mu \coloneqq \cev_{(C, \gamma)} \circ \tbeh_{(X,\xi)}$ is a marking of $(X,\xi)$ with $(C,\gamma)$. By Proposition \ref{prop:markingsAndPreruns} (ii), there exists a pre-run $(R,\rho = \langle \rho_F, \rho_X, \rho_Q \rangle,r_I)$ of $\algaut_{\mathsf{C}}$ on $(X,\xi,x_I)$ with $\rho_Q(r_I) \in \inj{1}[C]$ and $[\id,\prj{2}] \circ \rho_Q = \mu \circ \rho_X$.

        Hence $ [\id, \prj{2}] \circ \rho_Q = \mu \circ \rho_X = \cev_{(C, \gamma)} \circ \tbeh_{(X,\xi)} \circ \rho_X$.
        We conclude that $\rho_Q(r_I) \in (\inj{1} \circ \cev_{(C, \gamma)} \circ \tbeh_{(X, \xi)})(x_I) \in \inj{1}[U]$, i.e., $(R, \rho, r_I)$ is an accepting pre-run. By Lemma \ref{lem:preRunImpliesAccepted}, $(X, \xi, x_I)$ is accepted by $\algaut_{\mathsf C}$.
        
        \begin{equation*}
            \begin{tikzcd}
        	R & X && \thin Z \\
        	Q &&& C
        	\arrow["{\rho_X}", from=1-1, to=1-2]
        	\arrow["{\rho_Q}"', from=1-1, to=2-1]
        	\arrow["{\tbeh_{(X, \xi)}}", from=1-2, to=1-4]
        	\arrow["\mu"', from=1-2, to=2-4]
        	\arrow["{\cev_{(C, \gamma)}}", from=1-4, to=2-4]
        	\arrow["{[\id,\prj{2}]}"', from=2-1, to=2-4]
            \end{tikzcd}
        \end{equation*}

    Conversely, suppose there exists an accepting run $(R, \rho = \langle \rho_F, \rho_X, \rho_Q \rangle ,r_I)$ of $\algaut_{\mathsf{C}}$ on $(X, \xi, x_I)$. By Proposition \ref{prop:markingsAndPreruns} (i), the map $\mu = [\id, \prj{2}] \circ \rho_Q$ is a marking of $(R, \rho_F)$ with $(C, \gamma)$. By Propositions \ref{prop:existenceMarkings} and \ref{prop:uniquenessOfMarkings}, we have $\mu = \cev_{(C, \gamma)} \circ \tbeh_{(X, \xi)}$. Hence:

    \vspace{-10px}
    \begin{minipage}{0.6\textwidth}
        \begin{multline*}
            (\cev_{(C, \gamma)} \circ \tbeh_{(X, \xi)})(x_I) = (\cev_{(C, \gamma)} \circ \tbeh_{(X, \xi)} \circ \rho_X)(r_I) = \\
            (\cev_{(C, \gamma)} \circ \tbeh_{(R, \rho_F)})(r_I) = 
            \mu(r_I) = ([\id, \prj{2}] \circ \rho_Q)(r_I).
        \end{multline*}
    \end{minipage}
    \begin{minipage}{0.4\textwidth}
        \begin{equation*}
            \begin{tikzcd}
                R && C \\
                \\
                X && \thin Z
                \arrow["\mu", from=1-1, to=1-3]
                \arrow["{\rho_X}"', from=1-1, to=3-1]
                \arrow["{\tbeh_{(R, \rho_F)}}"{description}, from=1-1, to=3-3]
                \arrow["{\tbeh_{(X, \xi)}}"', from=3-1, to=3-3]
                \arrow["{\cev_{(C, \gamma)}}"', from=3-3, to=1-3]
            \end{tikzcd}
        \end{equation*}
    \end{minipage}
    But since $(R, \rho, r_I)$ is accepting, $\rho_Q(r_I) \in \inj{1}(U)$, so $(\cev_{(C, \gamma)} \circ \tbeh_{(X, \xi)})(x_I) = ([\id, \prj{2}] \circ \rho_Q)(r_I) \in U$.
\end{proof}

In order to prove unambiguity of the algebraic automaton, we first show that in every pre-run $\rho_Q$ is uniquely determined by $\rho_F$ and $\rho_X$.

\begin{lemma}
\label{lem:preRunUniquelyDetermined}
    Let $\mathsf C$ be a coherent $(F+G)$-algebra with a recognising set, and $(X, \xi, x_I)$ be a thin pointed $F$-coalgebra.
    If $(R, \langle \rho_F, \rho_X, \rho_Q \rangle, r_I)$ and $(R, \rho_F, \rho_X, \rho_Q' \rangle, r_I)$ are pre-runs of $\algaut_{\mathsf C}$ on $(X, \xi, x_I)$, then $\rho_Q = \rho_Q'$.
\end{lemma}
\begin{proof*}{Proof (Sketch)}
    Let $\mathsf C = (C, \gamma, U)$. By Proposition \ref{prop:markingsAndPreruns} (i), $[\id,\prj{2}] \circ \rho_Q: R \to C$ is a marking of $(R, \rho_F)$ with $(C, \gamma)$. By Propositions \ref{prop:existenceMarkings} and \ref{prop:uniquenessOfMarkings}, $[\id,\prj{2}] \circ \rho_Q = \cev_{(C, \gamma)} \circ \tbeh_{(X, \xi)}$. Similarly, $[\id,\prj{2}] \circ \rho_Q' = \cev_{(C, \gamma)} \circ \tbeh_{(X, \xi)}$, so $[\id,\prj{2}] \circ \rho_Q = [\id,\prj{2}] \circ \rho_Q'$. Now it can be shown by induction on the successor relation of $(R,\rho_F, r_I)$ that $\rho_Q = \rho_Q'$.
\end{proof*}

\begin{theorem}
\label{thm:algebraicAutomatonUnambiguity}
    Let $\mathsf C$ be a coherent $(F+G)$-algebra with a recognising set, and $(X,\xi,x_I)$ be a thin pointed $F$-coalgebra. The algebraic automaton $\algaut_{\mathrm{C}}$ is unambiguous on $(X,\xi,x_I)$.
\end{theorem}
\begin{proof}
    Let $(R', \rho' = \langle \rho_F', \rho_X', \rho_Q' \rangle, r_I')$ and $(R'', \rho'' = \langle \rho_F'', \rho_X'', \rho_Q'' \rangle, r_I'')$ be two accepting runs of $\algaut_{\mathsf{C}}$ on $(X, \xi, x_I)$. We prove equality between these runs by exhibiting a span between them.

    \vspace{-10px}
    \noindent\begin{minipage}{0.53\textwidth}
        Since $\rho_X': (R', \rho_F', r_I') \to (X,\xi, x_I)$ and $\rho_X'': (R'', \rho_F'', r_I'') \to (X,\xi, x_I)$ are pointed $F$-coalgebra morphisms and $F$ preserves weak pullbacks, there exists a reachable pointed $F$-coalgebra $(R, \rho_F, r_I)$ with $F$-coalgebra morphisms $f': (R, \rho_F, r_I) \to (R', \rho_F', r_I')$ and $f'': (R, \rho_F, r_I) \to (R'', \rho_F'', r_I'')$ such that $\rho_X' \circ f' = \rho_X'' \circ f''$.
    \end{minipage}
    \begin{minipage}{0.47\textwidth}
        \vspace{-5px}
        \begin{equation*}
            \begin{tikzcd}
        	& {(R',\rho_F',r_I')} \\
        	{(R,\rho_F,r_I)} && {(X,\xi,x_I)} \\
        	& {(R'',\rho_F'',r_I'')}
        	\arrow["{\rho_X'}", from=1-2, to=2-3]
        	\arrow["{f'}", from=2-1, to=1-2]
        	\arrow["{f''}"', from=2-1, to=3-2]
        	\arrow["{\rho_X''}"', from=3-2, to=2-3]
            \end{tikzcd}
        \end{equation*}
    \end{minipage}
    
    Define $\rho_X \coloneqq \rho_X' \circ f'$.  It can be shown that $(R,\langle \rho_F,\rho_X, \rho_Q' \circ f' \rangle, r_I)$ and $(R,\langle \rho_F,\rho_X, \rho_Q'' \circ f'' \rangle, r_I)$ are pre-runs,
    so according to Lemma \ref{lem:preRunUniquelyDetermined}, $\rho_Q' \circ f' = \rho_Q'' \circ f''$. Hence, by setting $\rho_Q \coloneqq \rho_Q' \circ f'$, we have that $f': (R,\langle \rho_F,\rho_X, \rho_Q \rangle, r_I) \to (R',\rho',r_I')$ and $f'': (R,\langle \rho_F,\rho_X, \rho_Q \rangle, r_I) \to (R'',\rho'',r_I'')$ are $(F \times \Delta_X \times \Delta_Q)$-morphisms.
    Consequently, $((R,\rho,r_I), f',f'')$ is a span in the category of pointed $(F \times \Delta_X \times \Delta_Q)$-coalgebras. This means that $(R', \rho', r_I')$ and $(R'', \rho'', r_I'')$ are behaviourally equivalent pointed subcoalgebras of the final coalgebra, hence they are equal.
\end{proof}

\section{Combining the Two Constructions}\label{sec:mainResults}
Here we derive the main results of the paper, by employing the automaton algebra and the algebraic automaton constructions. We begin by showing that, when restricted to thin coalgebras, every parity $F$-automaton has an equivalent unambiguous parity $F$-automaton. We will make use of a property that we hinted at earlier: that rational algebras induce automata with $\omega$-regular acceptance.

\begin{lemma}
\label{lem:algebraicAutomatonOfRationalIsOmegaRegular}
    Let $\mathsf C = (C, \gamma, U)$ be a rational $(F+G)$-algebra with a recognising set. Then the acceptance condition of the algebraic automaton $\algaut_{\mathsf{C}}$ is $\omega$-regular.
\end{lemma}
\begin{proof*}{Proof (Sketch)}
    It follows from the $\omega$-semigroup structure on $(C,\gamma)$ that, for every $c \in C$, the language $L(c) \coloneqq \{ (\bar c'_n)_{n\in\omega} \in (F'C)^\omega \mid \gamma_1((\bar c'_n)_{n\in\omega}) = c \}$ is $\omega$-regular. This can be used to show that the acceptance condition $\Acc = \{ \inj{1}(c_0) \cdot (\inj{2}(\bar c'_n, c_n))_{n > 0} \mid \forall m (c_m = \gamma_1((\bar c'_n)_{n > m})) \}$ is $\omega$-regular.
\end{proof*}

\begin{theorem}
\label{thm:unambiguousParityAutomaton}
    For every parity $F$-automaton $\mathcal A$, there exists a parity automaton $\mathcal A'$ such that:
    \begin{enumerate}
        \item[(i)] $\mathcal A$ and $\mathcal A'$ accept the same thin $F$-coalgebras, and
        \item[(ii)] $\mathcal A'$ is unambiguous on thin $F$-coalgebras.
    \end{enumerate}
\end{theorem}
\begin{proof}
    Since parity conditions are prefix-agnostic, by Theorem \ref{thm:automatonAlgebraLanguage}, the automaton algebra $\autalg_{\mathcal A}$ accepts exactly those thin behaviours accepted by $\mathcal A$. By Proposition \ref{prop:automatonAlgebraRational}, $\autalg_{\mathcal A}$ is rational, so, by Lemma \ref{lem:algebraicAutomatonOfRationalIsOmegaRegular}, its algebraic automaton $\mathcal B \coloneqq \algaut_{\autalg_{\mathcal A}}$ has an $\omega$-regular acceptance condition. By Theorem \ref{thm:algebraicAutomatonLanguage}, $\mathcal B$ accepts the same thin coalgebras as $\mathcal A$, while by Theorem \ref{thm:algebraicAutomatonUnambiguity}, $\mathcal B$ is unambiguous on thin coalgebras. Finally, applying Proposition \ref{prop:omegaRegularToParity} to $\mathcal B$ gives us the desired automaton $\mathcal A'$.
\end{proof}

As our second main result, we give an automaton-theoretic characterisation of languages of finite coherent $(F+G)$-algebras. Concretely, we show that coherent algebras are as expressive as automata with prefix-agnostic acceptance (restricted to thin coalgebras). The key observation is that the acceptance condition of the algebraic automaton can be adjusted to a prefix-agnostic condition.

\begin{lemma}
\label{lem:algebraicAutomatonPrefixAgnostic}
    Let $\mathsf C = (C, \gamma, U)$ be a finite coherent $(F+G)$-algebra with a recognising set. There exists an $F$-automaton with a prefix-agnostic acceptance condition whose runs coincide with the runs of $\algaut_{\mathsf C}$.
\end{lemma}
\begin{proof*}{Proof (Sketch)} \qedoverwrite
    Let $\algaut_{\mathsf C} = (Q, \delta, Q_I, \Acc)$. Using coherence of $\mathsf C$, the following automaton can be shown to satisfy the desired conditions: $\mathcal A' \coloneqq (Q, \delta, Q_I, \Acc')$ with:
    \begin{multline*}
        \Acc' \coloneqq \{ (q_n)_{n \in \omega} \in Q^\omega \mid \exists m \in \omega, (\bar c'_n)_{n \geq m} \in F'C, (c_n)_{n \geq m} \in C^\omega: \\
        (q_n)_{n \geq m} = (\inj{2}(\bar c_n', c_n))_{n \geq m}  \land
        \forall k \geq m (c_m = \gamma_1((\bar c'_n)_{n \geq k})) \}. \qedhere
    \end{multline*}
\end{proof*}

\begin{theorem}
\label{thm:coherentAlgebrasLanguageChar}
    Restricted to thin $F$-coalgebras, finite coherent $(F+G)$-algebras recognise exactly the languages accepted by $F$-automata with a prefix-agnostic acceptance condition. More precisely, a language $L$ of thin $F$-behaviours equals $L(\mathsf{C})$, for some finite coherent $(F+G)$-algebra $\mathsf{C}$, if and only if $L$ consists of those thin $F$-behaviours accepted by $\mathcal A$, for some $F$-automaton $\mathcal A$ with prefix-agnostic acceptance. 
\end{theorem}
\begin{proof}
    Given a finite coherent $(F+G)$-algebra $\mathsf C$ with a recognising set, we have by Theorem \ref{thm:algebraicAutomatonLanguage} that the language of $\mathsf C$ consists precisely of the thin coalgebras accepted by its algebraic automaton. By Lemma~\ref{lem:algebraicAutomatonPrefixAgnostic}, there exists an equivalent prefix-agnostic automaton. Conversely, for every $F$-automaton with prefix-agnostic acceptance, by Theorem \ref{thm:automatonAlgebraLanguage}, its automaton algebra recognises precisely those thin coalgebras accepted by the automaton.
\end{proof}

\section{Conclusion}\label{sec:conclusion}
In this paper, we saw how to connect $F$-automata with prefix-agnostic acceptance to finite coherent $(F+G)$-algebras in order to transform an arbitrary $F$-automaton into an unambiguous one. We gave two constructions: the automaton algebra and the algebraic automaton constructions, both of which generalise the corresponding classical constructions for thin trees~\cite{Skrzypczak2016}. In order to prove unambiguity of the algebraic automaton (Theorem \ref{thm:algebraicAutomatonUnambiguity}), we linked algebraic automaton pre-runs to markings, which are $F$-coalgebra-to-algebra morphisms with an extra condition. We used the inductive structure of thin behaviours~\cite{ChernevCirsteaHansenKupkeThinCoalg} to show existence and uniqueness of markings, which implied existence and uniqueness of runs.
We concluded from the two constructions that finite $(F+G)$-algebras recognise the same languages of thin $F$-behaviours as $F$-automata with prefix-agnostic acceptance (Theorem \ref{thm:coherentAlgebrasLanguageChar}).

In applications, one usually considers parity $F$-automata. Hence we showed that the unambiguous automaton obtained from a parity $F$-automaton is itself a parity $F$-automaton (Theorem \ref{thm:unambiguousParityAutomaton}). To this end, we identified rational $(F+G)$-algebras as a subclass of coherent algebras that correspond to parity $F$-automata. We observed that rational algebras generalise thin algebras~\cite{Skrzypczak2016} to analytic functors.

In addition to providing a useful generalisation (beyond trees and polynomial functors) of an existing construction, our use of the context decomposition operator in the definition of the algebraic automaton, and our key insight that markings (as defined in \cite{Skrzypczak2016}) correspond to $(F+G)$-coalgebra-to-algebra morphisms shed new light on the original construction in loc.~cit.~and on the reasons it delivers unambiguity.

A natural direction for future work is to incorporate our unambiguous parity
automaton construction into model-checking algorithms,
such as the one proposed in \cite{CirsteaKupke:CSL2023}.
In this context, the size of the resulting automaton is crucial. In principle, 
our constructions yield at least
an exponential blow-up but simple optimisations such as removing unreachable states could considerably improve the automaton size.

We defined recognition by coherent algebra only for languages of thin $F$-behaviours, but some of our constructions can be extended to all $F$-behaviours. In particular, by considering $(F+G)$-algebra morphisms from $(Z,\beta)$ to a finite coherent $(F+G)$-algebra, we obtain a notion of recognition for arbitrary $F$-behaviours. We can extend Theorem \ref{thm:automatonAlgebraLanguage} to show that the automaton algebra recognises the same language over all $F$-behaviours.
In contrast, the properties of the algebraic automaton make essential use of thinness and, without it, neither Theorem \ref{thm:algebraicAutomatonLanguage} nor Theorem \ref{thm:algebraicAutomatonUnambiguity} seem to hold. 
We leave further investigations into coherent algebra recognition of non-thin behaviours as future work.

Finally, our characterisation of languages of finite coherent $(F+G)$-algebras, together with Example~\ref{ex:automatonNonRegularLanguage}, show that the expressivity of these algebras lies beyond regular languages. We ``corrected'' this by equipping the algebras with additional structure, thus obtaining rational algebras. Yet, rationality only played a role in ensuring that the algebraic automaton has a parity condition. This suggests that coherent algebras could also be specialised with alternative additional structure in order to study different classes of languages, while maintaining the correctness of the unambiguity construction.

%\pagebreak

%\bibliographystyle{./entics}
%\bibliography{references}

%\pagebreak

\appendix
\section*{Appendix}
\section{Detailed Proofs from Section \ref{sec:preliminaries}}

\begin{proof*}{Proof of Proposition \ref{prop:trigNaturalitySquareIsWeakPullback}}
    Let $f: X \to Y$. We are to show that the following diagram is a weak pullback.
    \begin{equation*}
        \begin{tikzcd}
            {F'X \times X} & FX \\
            {F'Y \times Y} & FY
            \arrow["{\trig_X}", from=1-1, to=1-2]
            \arrow["{F'f\times f}"', from=1-1, to=2-1]
            \arrow["Ff", from=1-2, to=2-2]
            \arrow["{\trig_Y}"', from=2-1, to=2-2]
        \end{tikzcd}
    \end{equation*}
    It suffices to show that, for all $\bar x \in FX$, $\bar y' \in F'Y$ and $y \in Y$ with $\trig_Y(\bar y', y) = Ff(\bar x)$, there exist $\bar x' \in F'X$ and $x \in X$ such that $F'f(\bar x') = \bar y'$, $f(x) = y$ and $\trig_X(\bar x', x) = \bar x$. Let $\bar x = (i, [\phi]_{H_i})$ for some $i \in I, \phi: U_i \to X$ and $\bar y' = (j, [u, \psi]_{H_j})$ for some $j \in I$, $u \in U_j$ and $\psi: U_j \setminus \{ u \} \to Y$. Since:
    \begin{equation*}
        (j, [\psi \cup \{ \langle u, y \rangle \}]_{H_j}) = \trig_Y(\bar y', y) = Ff(\bar x) = Ff(i, [\phi]_{H_i}) = (i, [f \circ \phi]_{H_i}),
    \end{equation*}
    we have $i = j$ and, without loss of generality, $y = f(\phi(u))$ and $\psi(v) = f(\phi(v))$ for all $v \in U_i \setminus \{ u \}$ (otherwise, we could take different representatives $\phi$ and $\psi$). Now take $\bar x' \coloneqq (i, [u, \chi]_{H_i})$ with $\chi = \phi \setminus \{ \langle u, \phi(u) \rangle \}$ and $x \coloneqq \phi(u)$. We have $F'f(\bar x') = \bar y'$, $f(x) = y$, $f(x) = y$ and $\trig(\bar x', x) = \bar x$, as desired.
\end{proof*}

\begin{proof*}{Proof of Lemma \ref{lem:decompProperties}}
    (i). Suppose $(i, [\phi]_{H_i}) \in FX$ and $\decomp_X(i, [\phi]_{H_i}) = (i, [\psi]_{H_i})$, where $\psi$ is as in Definition \ref{def:decomp}. Then
    $(F\prj{2} \circ \decomp_X)(i, [\phi]_{H_i}) = F\prj{2}(i, [\psi]_{H_i}) = (i, [\prj{2} \circ \psi]_{H_i}) = (i, [\phi]_{H_i})$.

    (ii). Let $Y \coloneqq F'X \times X$ and suppose $(\bar y', y) \in F'Y \times Y$, with $\bar y' = (i, [u, \chi]_{H_i})$ for $u \in U_i$ and $\chi: U_i \to Y$, satisfies $\trig_Y(\bar y', y) \in \decomp[FX]$. This means that
    $(i, [\chi \cup \{ \langle u, y \rangle \}]_{H_i}) = \trig(\bar y', y) = \decomp(j, [\phi]_{H_j}) = (j, [\psi]_{H_j})$
    for some $(j, [\phi]_{H_j}) \in FX$ and $\psi$ as in Definition \ref{def:decomp}. Hence $i = j$ and $[\chi \cup \{ \langle u, y \rangle \}]_{H_i} = [\psi]_{H_i})$. Without loss of generality, we can assume that $(u, \chi)$ is chosen in such a way that $\chi \cup \{ \langle u, y \rangle \} = \psi$. It follows that $\bar y' = \psi \setminus \{ u, \psi(u) \rangle \}$ and $y = \psi(u)$. Now:
    \begin{equation*}
        F'\prj{2}(\bar y') = (i, [u, \prj{2} \circ \chi]_{H_i}) = (i, [u, \phi \setminus \{ \langle u, \phi(u) \rangle ]_{H_i}) =
        \prj{1}((i, [u, \phi \setminus \{ \langle u, \phi(u) \rangle \}]_{H_i}), \phi(u)) = \prj{1}(y).
    \end{equation*}

    (iii). Suppose $(i, \phi) \in FX$, $\decomp_X(i, [\phi]_{H_i}) = (i, [\psi]_{H_i})$, with $\psi$ as in Definition \ref{def:decomp}, and $(\bar x', x) \in \Base_F(i, [\psi]_{H_i})$. There exists a $u \in U_i$ such that $\bar x' = (i, [u, \phi \setminus \{ \langle u, \phi(u) \rangle \}]_{H_i})$ and $x = \phi(u)$, hence
    $\trig_X(\bar x', x) = (i, [(\phi \setminus \{ \langle u, \phi(u) \rangle \}) \cup \{ \langle u, \phi(u) \rangle \}]_{H_i}) = (i, \phi)$.
\end{proof*}

\section{Detailed Proofs from Section \ref{sec:coalgebraicAutomata}}

The following lemma is used in the proof of Lemma \ref{lem:preRunImpliesAccepted} and Proposition \ref{prop:acceptanceInvariantUnderMorphisms}.
\begin{lemma}
\label{lem:preservationReflectionofPreRunProperties}
    Let $\mathcal A = (Q, \delta, Q_I, \Acc)$ be an $F$-automaton, $(X, \xi, x_I)$ be a pointed $F$-coalgebra, $\sfR = (R, \rho = \langle \rho_F, \rho_X, \rho_Q \rangle, r_I)$, $\sfR' = (R', \rho' = \langle \rho_F', \rho_X', \rho_Q' \rangle, r_I')$ be pointed $(F\times \Delta_X \times \Delta_Q)$-coalgebras, and $f: (R, \rho_F, r_I) \to (R', \rho_F', r_I')$ be a pointed $F$-coalgebra morphism.
    \begin{enumerate}
        \item[(a)] Suppose $\rho_X = \rho_X' \circ f$. If $\sfR'$ satisfies property (i) of pre-runs then so does $\sfR$.  If $f$ is an epi then the converse holds.
        \item[(b)] Suppose $\rho_Q = \rho_Q' \circ f$. 
         If $\sfR'$ satisfies property (ii) of pre-runs then so does $\sfR$; and similarly for property (iii) of pre-runs. 
         If $f$ is an epi then both converse statements hold.
    \end{enumerate}
\end{lemma}
\begin{proof}
    \textit{Part (a), first claim.} Assume that $(R', \rho', r_I')$ satisfies property (i) of pre-runs and show that so does $(R, \rho, r_I)$.
    This amounts to showing commutativity of $(1)$ and $(3)$ in the diagrams below.
    The small triangles commute by the assumption $\rho_X = \rho_X' \circ f$;
    $(2)$ and $(4)$ commute by property (i) for $(R',\rho',r'_I)$.
    In the diagram on the left, 
    the outer paths commutes because $f$ preserves the coalgebra root. Hence $(1)$ also commutes. In the diagram on the right, the outer paths commute by the assumption and functoriality of $F$. The top crescent commutes, since $f$ is an $F$-coalgebra morphism. It follows that $(3)$ commutes.
    \begin{equation*}
    \SelectTips {cm}{}
    
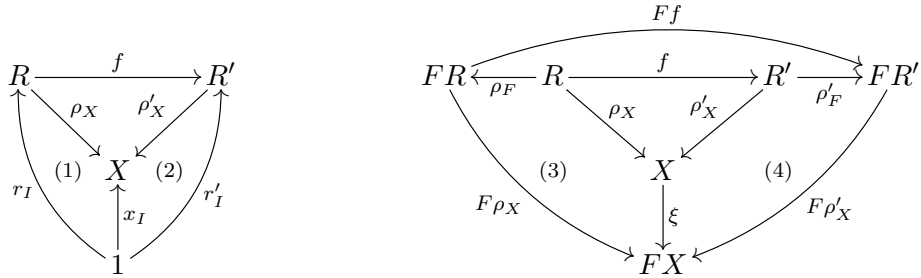
\xymatrix{
            R \ar[rr]^-{f} \ar[dr]^-{\rho_X } && {R'} \ar[dl]_-{\rho'_X } \\
            & X \\
            & 1 \ar[u]_-{x_I} \ar@/^1pc/[uul]^-{r_I} \ar@/_1pc/[uur]_-{r'_I} \ar@{}[uul]|{(1)} \ar@{}[uur]|{(2)}
    }
        \qquad\qquad\qquad
        \SelectTips {cm}{}
    \xymatrix{  
            %R \ar[rr]^-{f} \ar[dr]^-{\rho_X } && {R'} \ar[dl]_-{\rho'_X } \\
            %& X \\
            %& 1 \ar[u]_-{x_I} \ar@/^1pc/[uul]^-{r_I} \ar@/_1pc/[uur]_-{r'_I} \ar@{}[uul]|{(i)} \ar@{}[uur]|{(ii)}
            FR \ar@/^1.5pc/[rrrr]^-{Ff} \ar@/_1pc/[ddrr]_-{F\rho_X} & R \ar[l]^-{\rho_F} \ar[rr]^-{f} \ar[dr]^-{\rho_X} && R' \ar[dl]_-{\rho'_X}  \ar[r]_-{\rho'_F}& FR' \ar@/^1pc/[ddll]^-{F\rho'_X} \\
            && X \ar[d]^-{\xi} &&\\
            && FX \ar@{}[uull]|{(3)} \ar@{}[uurr]|{(4)} &&
    }    
    \end{equation*}
    \textit{Part (a), converse claim.} Assume $f$ is epic and $(R, \rho, r_I)$ satisfies property (i) of pre-runs, i.e., in the above diagrams, $(1)$ and $(3)$ commute. Show that $(2)$ and $(4)$ commute.
    Commutativity of $(2)$ follows by an easy diagram chase. 
    It suffices to show commutativity of $(4)$ precomposed with the epimorphism $f$. 
    This again follows by an easy diagram chase, using commutativity of the other parts of the diagram.

\noindent\begin{minipage}{0.6\textwidth}
    \textit{Part (b), property (ii).} Consider the diagram to the right. 
    The inner and outer triangles commute by the assumption $\rho_Q = \rho_Q' \circ f$ and functoriality, and the top crescent commutes since $f$ is an $F$-coalgebra morphism.
    Hence for all $r \in R$:\\
    $(F\rho_Q \circ \rho_F)(r) \in (\delta \circ \rho_Q)(r) \iff (F\rho_Q' \circ \rho_F' \circ f)(r) \in (\delta \circ \rho_Q' \circ f)(r)$.\\
    It follows that if $(R', \rho', r_I')$ satisfies property (ii) of pre-runs, then so does $(R, \rho, r_I)$. 
    For the converse, we use the same equivalence and the fact that if $f$ is epic, then every $r' \in R'$ is of the form $f(r)$ for some $r \in R$.
    \end{minipage}
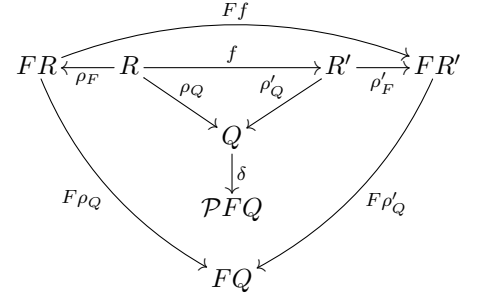
    \begin{minipage}{0.46\textwidth}
    \begin{equation*}
    \hspace{-2em}
     \scalebox{.9}{       \SelectTips {cm}{}
    \xymatrix@R=1.5em@C=2em{  
            FR \ar@/^1.5pc/[rrrr]^-{Ff} \ar@/_1pc/[dddrr]_-{F\rho_Q} & R \ar[l]^-{\rho_F} \ar[rr]^-{f} \ar[dr]^-{\rho_Q} && R' \ar[dl]_-{\rho'_Q}  \ar[r]_-{\rho'_F}& FR' \ar@/^1pc/[dddll]^-{F\rho'_Q} \\
            && Q \ar[d]^-{\delta} &&\\
            && \Pow FQ &&\\
            && FQ
    }}
    \end{equation*}
    \end{minipage}

    \textit{Part (b), property (iii).} Suppose $(R', \rho', r_I')$ satisfies property (iii) of pre-runs and let $(r_n)_{n\in\omega} \in R^\omega$ be such that $r_0 = r_I$ and $\forall n(r_{n+1} \in \Base_F(\rho_F(r_n))$. Since $f$ is a pointed coalgebra morphism, $(f(r_n))_{n\in\omega}$ satisfies $f(r_0) = r_I'$ and $\forall n(f(r_{n+1}) \in \Base_F(\rho_F'(f(r_n))))$. Hence $(\rho_Q(r_n))_{n\in\omega} = (\rho_Q'(f(r_n)))_{n\in\omega} \in \Acc$.

        Conversely, suppose $(R, \rho, r_I)$ satisfies property (iii) of pre-runs. Let $(r_n')_{n\in\omega} \in (R')^\omega$ be such that $r_0' = r_I'$ and $\forall n(r_{n+1}' \in \Base_F(\rho_F'(r_n')))$. Define a sequence $(r_n)_{n \in \omega} \in R^\omega$ inductively as follows: $r_0 = r_I$, for each $n \in \omega$, take $r_{n+1}$ with $f(r_{n+1}) = r_{n+1}'$ and $r_{n+1} \in \Base_F(\rho(r_n))$. The latter is possible, because $f$ is epic and an $F$-coalgebra morphism. Now
            $(\rho_Q'(r_n'))_{n\in\omega} = (\rho_Q'(f(r_n)))_{n\in\omega} = (\rho_Q(r_n))_{n\in\omega} \in \Acc$.
\end{proof}

\begin{proof*}{Proof of Lemma \ref{lem:preRunImpliesAccepted}}
    Let $(R, \rho, r_I)$ be an accepting pre-run of $\mathcal A$ on $(X, \xi, x_I)$ and $f: (R, \rho, r_I) \twoheadrightarrow (R', \rho', r_I')$ be the $(F \times \Delta_X \times \Delta_Q)$-coalgebra morphism mapping $(R, \rho, r_I)$ onto its image $(R', \rho', r_I')$ in the final coalgebra. By reachability of $(R, \rho, r_I)$, it follows that $(R', \rho', r_I')$ is also reachable. Properties (ii) and (iii) hold in $(R', \rho', r_I')$ by Lemma \ref{lem:preservationReflectionofPreRunProperties}, so $(R', \rho', r_I')$ is a run. Finally, $(R, \rho, r_I)$ is accepting, so $(R, \rho', r_I')$ is also accepting.
\end{proof*}

\begin{proof*}{Proof of Proposition \ref{prop:acceptanceInvariantUnderMorphisms}}
    $(\Longrightarrow)$ Let $\mathsf R = (R, \rho = \langle \rho_F, \rho_X, \rho_Q \rangle, r_I)$ be an accepting run of $\mathcal A$ on $(X, \xi, x_I)$. Consider $\mathsf{R}' \coloneqq (R, \langle \rho_F, f \circ \rho_X, \rho_Q \rangle, r_I)$. Since $f$ and $\rho_X$ are pointed $F$-coalgebra morphisms, it follows that $\mathsf{R}'$ satisfies property (i) of pre-runs. Properties (ii), (iii) and acceptance are automatically satisfied, because $\mathsf R$ satisfies them. Therefore $\mathsf R'$ is an accepting pre-run on $(X', \xi', x_I')$. By Lemma \ref{lem:preRunImpliesAccepted}, we conclude that $\mathcal A$ accepts $(X', \xi', x_I')$.

    $(\Longleftarrow)$ Let $(R', \rho' = \langle \rho_F', \rho_X', \rho_Q' \rangle, r_I')$ be an accepting run of $\mathcal A$ on $(X', \xi', x_I')$. Let $(R \subseteq X \times R', \rho_X: R \to X, g: R \to R')$ be the pullback of $f: X \to X'$ and $\rho'_X: R' \to X'$. Since $F$ preserves weak pullbacks, there exists an $F$-coalgebra structure $\rho_F: R \to FR$ such that $\rho_X: (R, \rho_F) \to (X, \xi)$ and $g: (R, \rho_F) \to (R', \rho_F')$ are $F$-coalgebra morphisms. Since $f(x_I) = x_I' = \rho_X'(r_I')$, there exists $r_I = (x_I, r_I') \in R$ and we get the following commuting diagram of pointed $F$-coalgebras.
    \begin{equation*}
        \begin{tikzcd}
    	{(R,\rho_F,r_I)} & {(R',\rho_F',r_I')} \\
    	{(X,\xi,x_I)} & {(X',\xi',x_I')}
    	\arrow["g", from=1-1, to=1-2]
    	\arrow["{\rho_X}"', from=1-1, to=2-1]
    	\arrow["{\rho_X'}", from=1-2, to=2-2]
    	\arrow["f"', from=2-1, to=2-2]
        \end{tikzcd}
    \end{equation*}
    Consider $\mathsf R \coloneqq (R, \langle \rho_F, \rho_X, \rho_Q' \circ g \rangle, r_I)$. We have that the reachable part of $\mathsf R$ is a pre-run (by Lemma \ref{lem:preservationReflectionofPreRunProperties}) and, in addition, it is accepting. By Lemma \ref{lem:preRunImpliesAccepted}, $\mathcal A$ accepts $(X, \xi, x_I)$.
\end{proof*}

The next lemma gives an equivalent characterisation of unambiguity and will be used in the proof of Proposition \ref{prop:omegaRegularToParity}.

\begin{lemma}
\label{lem:unambiguityViaBehEquivalence}
    An $F$-automaton $\mathcal A$ is unambiguous on $\mathsf X = (X,\xi,x_I)$ if and only if every two accepting pre-runs of $\mathcal A$ on $\mathsf X$ are behaviourally equivalent.
\end{lemma}
\begin{proof}
    $(\Rightarrow)$ Suppose $\mathcal A$ is unambiguous on $\mathsf X$ and let $\mathsf R$, $\mathsf R'$ be two accepting pre-runs of $\mathcal A$ on $\mathsf X$. By Lemma \ref{lem:preRunImpliesAccepted}, the images of $\mathsf R$ and $\mathsf R'$ in the final coalgebra are accepting runs. Since $\mathcal A$ is unambiguous, these runs coincide. Therefore $\mathsf R$ and $\mathsf R'$ are behaviourally equivalent.

    $(\Leftarrow)$ Let $\mathsf R$, $\mathsf R'$ be two accepting runs of $\mathcal A$ on $\mathsf X$. Since $\mathsf R$ and $\mathsf R'$ are behaviourally equivalent subcoalgebras of the final coalgebra, they are equal.
\end{proof}

\begin{proof*}{Proof of Proposition \ref{prop:omegaRegularToParity}}
    Let $\mathcal A = (Q, \delta, Q_I, \Acc)$ be an $F$-automaton with an $\omega$-regular acceptance condition. It follows that there exists a deterministic parity word automaton $\mathcal A^w = (Q^w, \delta^w, q_I^w, \Omega^w)$ over the alphabet $Q$ that recognises the language $\Acc$. We define a parity $F$-automaton $\mathcal A' \coloneqq (Q', \delta', Q_I', \Omega')$:
    \begin{align*}
        &Q' \coloneqq Q \times Q^w, \qquad  Q_I' \coloneqq Q_I \times \{ q_I^w \}, \qquad \Omega'(q, q^w) \coloneqq \Omega^w(q^w), \\
        &\delta'(q, q^w) \coloneqq \{ F\langle\id, \con_Q^{Q^w}(\delta^w(q^w)(q))\rangle(\bar q) \mid \bar q \in \delta(q) \}, \quad\text{where} \quad \con_X^{Y}(y) \coloneqq (\lambda x. y) : X \to Y.
    \end{align*}
    For every pointed $F$-coalgebra $\mathsf X = (X, \xi, x_I)$, we show that $\mathcal A$ accepts $\mathsf X$ if and only if $\mathcal A'$ accepts $\mathsf X$, and that if $\mathcal A$ is unambiguous on $\mathsf X$, then $\mathcal A'$ is unambiguous on $\mathsf X$. Let $\PRun$, resp. $\PRun'$, denote the category where objects are accepting pre-runs of $\mathcal A$, resp. $\mathcal A'$, on $\mathsf{X}$ and arrows are pointed coalgebra morphisms.
    We define a mapping $K: \Ob(\PRun) \to \Ob(\PRun')$. Given $\mathsf{R} = (R, \langle \rho_F, \rho_X, \rho_Q \rangle, r_I) \in \Ob(\PRun)$, we define $K(\mathsf{R})$ to be the reachable part of $(R \times Q^w, \langle \rho_F', \rho_X', \rho_Q' \rangle, (r_I, q_I^w))$, where $\rho_F'(r, q^w) \coloneqq (F\langle\id, \con_R^{Q^w}(\delta^w(q^w)(\rho_Q(r)))\rangle \circ \rho_F)(r)$, $\rho_X' \coloneqq \rho_X \circ \prj{1}$ and $\rho_Q' \coloneqq \rho_Q \times \id$.
    Moreover, we define a mapping $M: \Ob(\PRun') \to \Ob(\PRun)$. Given an object $\mathsf{R}' = (R', \langle \rho_Q', \rho_X', \rho_Q' \rangle, r_I') \in \PRun'$, define $M(\mathsf{R}') \coloneqq (R', \langle \rho_F', \rho_X', \prj{1} \circ \rho_Q' \rangle, r_I')$.
    Finally, for all $\mathsf R \in \Ob(\PRun)$, $\mathsf R' \in \Ob(\PRun')$ and $f: \mathsf R \to M(\mathsf R') \in \Arr(\PRun)$, define a map $N(f): K(\mathsf R) \to \mathsf R' \in \Arr(\PRun')$ by $N(f) = f \circ \prj{1}$. Through lengthy but straightforward verifications, one shows that $K$, $M$ and $N$ are well-defined.
    
    We are now ready to show that $\mathcal A$ accepts $\mathsf X$ if and only if $\mathcal A'$ accepts $\mathsf{X}$, and that, assuming $\mathcal A$ is unambiguous on $\mathsf{X}$, we have that $\mathcal A'$ is unambiguous on $\mathsf X$. From the existence of $K$ and $M$, we know that $\PRun$ is non-empty if and only if $\PRun'$ is non-empty, i.e., $\mathcal A$ has an accepting pre-run on $\mathsf X$ if and only if $\mathcal A'$ has an accepting pre-run on $\mathsf X$. By Lemma \ref{lem:preRunImpliesAccepted}, this implies that $\mathcal A$ accepts $\mathsf X$ if and only if $\mathcal A'$ accepts $\mathsf X$. For unambiguity of $\mathcal A'$ on $\mathsf X$, by Lemma \ref{lem:unambiguityViaBehEquivalence}, it suffices to show that all $(\mathsf R^1)', (\mathsf R^2)' \in \PRun'$ are behaviourally equivalent. 
     Since $M((\mathsf R^1)')$ and $M((\mathsf R^2)')$ are accepting pre-runs and $\mathcal A$ is assumed unambiguous, by Lemma \ref{lem:unambiguityViaBehEquivalence} they are behaviourally equivalent. Since $F \times \Delta_X \times \Delta_Q$ preserves weak pullbacks, there exists a span $(f_1, f_2)$ as shown in the diagram below on the left. Now the (pushout of the) span $(N(f_1), N(f_2))$ below right witnesses that $(\mathsf R^1)'$ and $(\mathsf R^2)'$ are behaviourally equivalent.
    \begin{equation*}
        \begin{tikzcd}
            {M((\mathsf{R}^1)')} & \mathsf R & {M((\mathsf{R}^2)')} && {(\mathsf{R}^1)'} &{K(\mathsf R)}& {(\mathsf{R}^2)'}
            \arrow["{f_1}"', from=1-2, to=1-1]
            \arrow["{f_2}", from=1-2, to=1-3]
            \arrow["{N(f_1)}"', from=1-6, to=1-5]
            \arrow["{N(f_2)}", from=1-6, to=1-7]
        \end{tikzcd}
    \end{equation*}
\end{proof*}

\section{Detailed Proofs from Section \ref{sec:automatonToAlgebra}}

\begin{proof*}{Proof of Lemma \ref{lem:autAlgebraCoherent}}
    Let $\mathcal A = (Q, \delta, Q_I, \Acc)$ be an $F$-automaton, $\Acc$ be prefix-agnostic and $\autalg_{\mathcal A} = (C, [\gamma_0, \gamma_1], U)$ be the automaton algebra of $\mathcal A$.
    Suppose $(\bar c'_n)_{n \in \omega} \in (F'C)^\omega$. We show that
    $\gamma_1((\bar c'_n)_{n \in \omega}) = \gamma_0(\trig_C(\bar c'_0, \gamma_1((\bar c'_n)_{n > 0})))$
    by proving the two inclusions separately.

    $(\subseteq)$ Suppose $q \in \gamma_1((\bar c'_n)_{n \in \omega})$. By definition of $\gamma_1$, this means that there exist $(q_n)_{n \in \omega} \in \Acc$, $(\bar q'_n)_{n \in \omega} \in (F'Q)^\omega$ with $q_0 = q$, $\bar q'_n \inlift \bar c'_n$ and $\trig_Q(\bar q'_n, q_{n+1}) \in \delta(q_n)$ for all $n \in \omega$. By the definition of $\gamma_0$, it suffices to show that $\trig_Q(\bar q'_0, q_1) \inlift \trig_C(\bar c'_0, \gamma_1((\bar c'_n)_{n>0}))$. Let $W \coloneqq \{ (q, c) \in Q \times C \mid q \in c \}$. By definition of relation liftings, we need to find $\bar w \in FW$ with $F\prj{1}(\bar w) = \trig_Q(\bar q'_0, q_1)$ and $F\prj{2}(\bar w) = \trig_C(\bar c'_0, \gamma_1((\bar c'_n)_{n>0}))$. We will construct a suitable $\bar w$ of the form $\bar w = \trig_W(\bar w', w)$ for $\bar w' \in F'W$ and $w \in W$.
    
    Since $\Acc$ is prefix-agnostic, we have $(q_n)_{n > 0} \in \Acc$, therefore $q_1 \in \gamma_1((\bar c'_n)_{n > 0})$. Take $w \coloneqq (q_1, \gamma_1((\bar c'_n)_{n > 0})) \in W$. Moreover, by assumption, $\bar q'_0 \in \bar c'_0$, so there exists some $\bar w' \in F'W$ with $F'\prj{1}(\bar w') = \bar q'_0$ and $F\prj{2}(\bar w') = \bar c'_0$. We show that $\bar w \coloneqq (\bar w', w)$ satisfies the necessary condition:
    \begin{gather*}
        F\prj{1}(\bar w) = (F\prj{1} \circ \trig_W)(\bar w', w) = \trig_Q(F\prj{1}(\bar w'), \prj{1}(w)) = \trig_Q(\bar q'_0, q_1), \\
        F\prj{2}(\bar w) = (F\prj{2} \circ \trig_W)(\bar w', w) = \trig_C(F\prj{2}(\bar w'), \prj{2}(w)) = \trig_C(\bar c'_0, \gamma_1((\bar c'_n)_{n > 0})).
    \end{gather*}

    $(\supseteq)$ Suppose $q \in \gamma_0(\trig_C(\bar c'_0, \gamma_1(\bar c'_n)_{n > 0}))$. By definition of $\gamma_0$, there exists $\bar q \in FQ$ such that $\bar q \in \delta(q)$ and $\bar q \inlift \trig_C(\bar c'_0, \gamma_1((\bar c'_n)_{n > 0})) \eqqcolon \bar c$. Consider again $W$ as defined above. Since $\bar q \in \bar c$, there exists $\bar w \in FW$ with $F\prj{1}(\bar w) = \bar q$ and $F\prj{2}(\bar w) = \bar c$. By Proposition \ref{prop:trigNaturalitySquareIsWeakPullback}, the following diagram is a weak pullback:
    \begin{equation*}
        \begin{tikzcd}
            {F'W \times W} & FW \\
            {F'C \times C} & FC
            \arrow["{\trig_W}", from=1-1, to=1-2]
            \arrow["{F'\prj{2}\times\prj{2}}"', from=1-1, to=2-1]
            \arrow["{F\prj{2}}", from=1-2, to=2-2]
            \arrow["{\trig_C}"', from=2-1, to=2-2]
        \end{tikzcd}
    \end{equation*}
    Since $\trig_C(\bar c'_0, \gamma_1((\bar c'_n)_{n>0})) = \bar c = F\prj{2}(\bar w)$, there exists $(\bar w', w) \in F'W \times W$ with $F'\prj{2}(\bar w') = \bar c'_0$, $\prj{2}(w) = \gamma_1((\bar c'_n)_{n>0})$ and $\trig_W(\bar c', c) = \bar c$. Let $\bar q'_0 \coloneqq F'\prj{1}(\bar w')$ and $q_1 \coloneqq \prj{1}(w)$. Hence, we have $\bar q'_0 \inlift \bar c'_0$ and $q_1 \in \gamma_1((\bar c'_{n>0}))$. The latter means that there exist $(\bar q_n)_{n > 0} \in \Acc$ and $(\bar q'_n)_{n>0} \in (F'Q)^\omega$ with $\bar q'_n \inlift \bar c'_n$ and $\trig_Q(\bar q'_n, q_{n+1}) \in \delta(q_n)$ for all $n > 0$. Since $\Acc$ is prefix-agnostic, we have $(q_n)_{n \in \omega} \in \Acc$, so it suffices to show $\trig_Q(\bar q'_0, q_1) \in \delta(q_0)$:
        $\trig_Q(\bar q'_0, q_1) = \trig_Q(F'\prj{1}(\bar w'), \prj{1}(w)) = (F\prj{1} \circ \trig_W)(\bar w', w) = F\prj{1}(\bar w) = \bar q \in \delta(q_0)$.
\end{proof*}

\begin{proof*}{Proof of Theorem \ref{thm:automatonAlgebraLanguage}}
    Let $\mathcal A = (Q, \delta, Q_I, \Acc)$ and $\autalg_{\mathcal{A}} = (C, \gamma = [\gamma_0,\gamma_1], U)$. Define:
    \begin{align*}
        f : \thin Z \to C, \qquad
        z \mapsto \{ q \in Q \mid \text{there exists a run of $\mathcal A$ on $(\thin Z,\thin \zeta, z)$, starting from $q$}\}.
    \end{align*}
    We claim that $f: (\thin Z,\thin \beta) \to (C, \gamma)$ is an $(F+G)$-algebra morphism. Since $(\thin Z, \thin \beta)$ is an initial coherent $(F+G)$-algebra, this means $f = \cev_{(C, \gamma)}$. Now $\mathcal A$ accepts $(X, \xi, x_I)$ if and only if $\mathcal A$ accepts $(\thin Z, \thin \zeta, \tbeh_{(X, \xi)}(x_I))$ if and only if $(f \circ \tbeh_{(X, \xi)})(x) \cap Q_I \neq \emptyset$ if and only if $(\cev_{(C,\gamma)} \circ \tbeh_{(X,\xi)})(x) \in U$.

    In the remainder of the proof, we show that $f$ is indeed an $(F+G)$-algebra morphism. Firstly, we prove that $(f \circ \thin\beta_0)(\bar z) = (\gamma_0 \circ Ff)(\bar z)$, for an arbitrary $\bar z \in F\thin Z$, by considering the two inclusions separately.

    $(\subseteq)$ Suppose $q \in (f \circ \thin\beta_0)(\bar z)$, i.e., there exists a run of $\mathcal A$ on $(\thin Z, \thin \zeta, \thin\beta_0(z))$, starting from $q$. Let this run be denoted by $(R, \rho = \langle \rho_F, \rho_X, \rho_Q \rangle, r_I)$, where $X$ stands for $\thin Z$. We show that $\bar q \inlift \bar c$ for $\bar q \coloneqq (F\rho_Q \circ \rho_F)(r_I)$ and $\bar c \coloneqq Ff(\bar z)$.
    Define the relation $W \coloneqq \{ (q, c) \in Q \times C \mid q \in c \}$ and the map $g: R \to W$ with $g(r) \coloneqq (\rho_Q(r), (f \circ \rho_X)(r))$.
    This map is well-defined, because, for each $r \in R$, since $\Acc$ is prefix-agnostic, the reachable part of $(R, \rho, r)$ is a run of $\mathcal A$ on $(Z, \thin \zeta, \rho_X(r))$, starting from $\rho_Q(r)$. Now $\bar w \coloneqq (Fg \circ \rho_F)(r_I) \in FW$ witnesses $\bar q \inlift \bar c$, because $F\prj{1}(\bar w) = (F\rho_Q \circ \rho_F)(r_I) = \bar q$ and $F\prj{2}(\bar w) = (Ff \circ F\rho_X \circ \rho_F)(r_I) = (Ff \circ \thin \zeta \circ \rho_X)(r_I) = (Ff \circ \thin \zeta \circ \thin\beta_0)(\bar z) = Ff(\bar z) = \bar c$.
    Moreover, $\bar q \in (\delta \circ \rho_Q)(r_I) = (\delta \circ \thin\beta_0)(\bar z)$. Therefore $q \in \gamma_0(Ff(\bar z))$.

    $(\supseteq)$ Suppose $q \in (\gamma_0 \circ Ff)(\bar z)$. This means that there exists $\bar q \inlift Ff(\bar z)$ such that $\bar q \in \delta(q)$. Let $\bar w \in FW$ be a witness of $\bar q \inlift Ff(\bar z)$, i.e., $F\prj{1}(\bar w) = \bar q$ and $F\prj{2}(\bar w) = Ff(\bar z)$. For each $z \in \thin Z$ and $w \in W$ with $\prj{2}(w) = f(z)$, there exists a run $\mathsf{R}_{z, w}$ of $\mathcal A$ on $(Z, \thin z, z)$ starting at $\prj{1}(w)$. Let $\mathsf R = (R, \rho = \langle \rho_F, \rho_X, \rho_Q \rangle)$ be the disjoint union of all runs $R_{z, w}$, $h_Z: R \to \thin Z$ be the map sending each $r \in R_{z, w}$ to $z$ and $h_W: R \to W$ be the map sending each $r \in R_{z, w}$ to $w$. Note that $h_Z = \rho_X$ and $\prj{1} \circ h_W = \rho_Q$. The following left-hand side square is a weak pullback:
    \begin{equation*}
        \begin{tikzcd}
            R & W && FR & FW \\
            {\thin Z} & C && {F\thin Z} & FC
            \arrow["{h_W}", from=1-1, to=1-2]
            \arrow["{h_Z}"', from=1-1, to=2-1]
            \arrow["{\prj{2}}", from=1-2, to=2-2]
            \arrow["{Fh_W}", from=1-4, to=1-5]
            \arrow["{Fh_Z}"', from=1-4, to=2-4]
            \arrow["{F\prj{2}}", from=1-5, to=2-5]
            \arrow["f"', from=2-1, to=2-2]
            \arrow["Ff"', from=2-4, to=2-5]
        \end{tikzcd}
    \end{equation*}
    Since $F$ preserves weak pullbacks, the right-hand side square is also a weak pullback. Now, since $Ff(\bar z) = F\prj{2}(\bar w)$, there exists $\bar r \in FR$ such that $Fh_Z(\bar z) = \bar z$ and $Fh_W(\bar r) = \bar w$. We define a new pre-run by adding a fresh state $r_I$ to $\mathsf R$ with $\rho_F(r_I) \coloneqq \bar r$, $\rho_X(r_I) \coloneqq \thin\beta_0(\bar z)$ and $\rho_Q(r_I) \coloneqq q$,
    and taking the states in $\mathsf R$ reachable from $r_I$. In order to show that the resulting coalgebra is a pre-run, it suffices to show that properties (i) and (ii) of pre-runs hold at $r_I$, since $\mathsf R$ is a disjoint union of runs with one additional state $r_I$. Property (iii) of pre-runs holds automatically, because $\Acc$ is prefix-agnostic. We have:
    \begin{gather*}
        (F\rho_X \circ \rho_F)(r_I) = F\rho_X(\bar r) = Fh_Z(\bar r) = \bar z = (\thin \zeta \circ \thin\beta_0)(\bar z) = (\thin \zeta \circ \rho_X)(r_I), \\
        (F\rho_Q \circ \rho_F)(r_I) = F\rho_Q(\bar r) = F\prj{1} \circ h_W(\bar r) = F\prj{1}(\bar w) = \bar q \in \delta(q) = (\delta \circ \rho_Q)(r_I).
    \end{gather*}
    Now, by Lemma \ref{lem:preRunImpliesAccepted}, there exists a run of $\mathcal A$ on $(\thin Z, \thin \zeta, \thin \beta_0(\bar z))$, starting from $q$, hence $q \in (f \circ \thin\beta_0)(\bar z)$.

    Secondly, we prove $(f \circ \thin\beta_1)((\bar z'_n)_{n \in \omega}) = (\gamma_1 \circ Gf)((\bar z'_n)_{n \in \omega})$ for an arbitrary $(\bar z'_n)_{n \in \omega} \in G\thin Z$. We again consider the two inclusions separately.

    $(\subseteq)$ Suppose $q \in (f \circ \thin \beta_1)((\bar z'_n)_{n \in \omega})$, i.e., there exists a run $(R, \rho = \langle \rho_F, \rho_X, \rho_Q \rangle, r_I)$ of $\mathcal A$ on $(\thin Z, \thin \zeta, \thin \beta_1((\bar z'_n)_{n\in\omega})$, starting from $q$. We recursively define sequences $(r_n)_{n\in\omega} \in R^\omega$ and $(\bar r'_n)_{n\in\omega} \in GR$ satisfying $\rho_X(r_n) = \thin\beta_1((\bar z'_m)_{m \geq n})$, $F'\rho_X(\bar r'_n) = \bar z'_n$ and $\trig_R(\bar r'_n, r_{n+1}) = \rho_F(r_n)$.
    \begin{itemize}
        \item $r_0 \coloneqq r_I$.
        \item Suppose $r_n$ has been defined. By Proposition \ref{prop:trigNaturalitySquareIsWeakPullback}, we have the weak pullback:
        \begin{equation*}
            \begin{tikzcd}
                {F'R \times R} & FR \\
                {F'\thin Z \times \thin Z} & F\thin Z
                \arrow["{\trig_R}", from=1-1, to=1-2]
                \arrow["{F'\rho_X \times \rho_X}"', from=1-1, to=2-1]
                \arrow["{F\rho_X}", from=1-2, to=2-2]
                \arrow["{\trig_{\thin Z}}"', from=2-1, to=2-2]
            \end{tikzcd}
        \end{equation*}
        Since
            $\trig_{\thin Z}(\bar z'_n, \thin \beta_1((\bar z'_m)_{m \geq n+1})) = \thin\zeta(\thin \beta_1((\bar z'_m)_{m \geq n})) =
            (\thin\zeta \circ \rho_X)(r_n) = (F\rho_X \circ \rho_F)(r_n)$,
        we can take some $(\bar r'_n, r_{n+1}) \in F'R \times R$ satisfying $F'\rho_X(\bar r'_n) = \bar z'_n$, $\rho_X(r_{n+1}) = \thin \beta_1((\bar z'_m)_{m \geq n+1})$ and $\trig_R(\bar r'_n, r_{n+1}) = \rho_F(r_n)$.
    \end{itemize}
    Now define $q_n \coloneqq \rho_Q(r_n)$ and $\bar q'_n \coloneqq F'\rho_Q(\bar r'_n)$, for all $n \in \omega$. We show $(q_n)_{n\in\omega}$ and $(\bar q'_n)_{n\in\omega}$ witness $q \in \gamma_1((F'f(\bar z'_n))_{n\in\omega})$. By property (iii) of pre-runs, we get $(q_n)_{n\in\omega} \in \Acc$. The relation $\bar q'_n \in F'f(\bar z'_n)$ is witnessed by $g(r_n)$ ($g$ was defined earlier, in the $(\subseteq)$ inclusion for $(f \circ \thin\beta_0)(\bar z) = (\gamma_0 \circ Ff)(\bar z)$), because:
    \begin{gather*}
        F'\prj{1}(\bar w'_n) = (F'\prj{1} \circ F'g)(\bar r'_n) = F\rho_Q(\bar r'_n) = \bar q'_n, \\
        F'\prj{2}(\bar w'_n) = (F'\prj{2} \circ F'g)(\bar r'_n) = (F'f \circ F'\rho_X)(\bar r'_n) = F'f(\bar z'_n).
    \end{gather*}
    Lastly,
        $\trig_Q(\bar q'_n, q_{n+1}) = \trig_Q(F'\rho_Q(\bar r'_n), \rho_Q(r_{n+1})) = (F\rho_Q \circ \trig_R)(\bar r'_n, r_{n+1}) = (F\rho_Q \circ \rho_F)(r_n) \in (\delta \circ \rho_Q)(r_n) = \delta(q_n)$.

    $(\supseteq)$ Suppose $q \in (\gamma_0 \circ Gf)((\bar z'_n)_{n\in\omega})$. This means that there exist $(q_n)_{n\in\omega} \in \Acc$ and $(\bar q'_n)_{n\in\omega} \in GQ$ satisfying $\bar q'_n \inlift F'f(\bar z'_n)$ and $\trig_Q(\bar q'_n, q_{n+1}) \in \delta(q_n)$, for all $n \in \omega$. Let $\bar w'_n$ be a witness of $\bar q'_n \inlift F'f(\bar z'_n)$, i.e., $F'\prj{1}(\bar w'_n) = \bar q'_n$ and $F'\prj{2}(\bar w'_n) = F'f(\bar z'_n)$, for all $n \in \omega$. Consider again the runs $\mathsf R_{z, w}$ for each $z \in \thin Z$, $w \in W$, their disjoint union $\mathsf R$ and the maps $h_Z$, $h_W$ defined earlier in the proof (in the $(\supseteq)$ inclusion for $(f \circ \thin\beta_0)(\bar z) = (\gamma_0 \circ Ff)(\bar z)$). Since $F'$ preserves weak pullback, we have the following weak pullback:
    \begin{equation*}
        \begin{tikzcd}
            {F'R} & {F'W} \\
            {F'\thin Z} & {F'C}
            \arrow["{{F'h_W}}", from=1-1, to=1-2]
            \arrow["{{F'h_Z}}"', from=1-1, to=2-1]
            \arrow["{{F'\prj{2}}}", from=1-2, to=2-2]
            \arrow["F'f"', from=2-1, to=2-2]
        \end{tikzcd}
    \end{equation*}
    For each $n \in \omega$, we know $F'f(\bar z'_n) = F'\prj{2}(\bar w'_n)$, hence there exists $\bar r'_n \in F'R$ with $F'h_Z(\bar r'_n) = \bar z'_n$ and $F'h_W(\bar r'_n) = \bar w'_n$. We define a new pre-run by adding fresh states $\{ r_n \mid n \in \omega \}$ to $\mathsf R$ with $\rho_F(r_n) \coloneqq \trig_R(\bar r'_n, r_{n+1})$, $\rho_X(r_n) \coloneqq \thin\beta_1((\bar z'_m)_{m \geq n})$, $\rho_Q(r_n) \coloneqq q_n$, and taking the states reachable from $r_0$. In order to show that the resulting coalgebra is a pre-run, we verify that properties (i) and (ii) hold at every $r_n$ and that property (iii) holds. For property (i) of pre-runs, we have:
    \begin{multline*}
        (F\rho_X \circ \rho_F)(r_n) = (F\rho_X \circ \trig_R)(\bar r'_n, r_{n+1}) = \trig_{\thin Z}(F'\rho_X(\bar r'_n), \rho_X(r_{n+1})) = \\
        \trig_{\thin Z}(F'h_Z(\bar r'_n), \rho_X(r_{n+1})) = \trig_{\thin Z}(\bar z'_n, \thin\beta_1((z_m)_{m \geq n+1})) =
        (\thin \zeta \circ \thin\beta_1)((\bar z'_m)_{m \geq n}) = (\thin \zeta \circ \rho_X)(r_n).
    \end{multline*}
    For property (ii) of pre-runs, we have:
    \begin{multline*}
        (F\rho_Q \circ \rho_F)(r_n) = (F\rho_Q \circ \trig_R)(\bar r'_n, r_{n+1}) = \trig_Q(F'\rho_Q(\bar r'_n), \rho_Q(r_{n+1})) = \\
        \trig_Q((F'\prj{1} \circ F'h_W)(\bar r'_n), q_{n+1}) = \trig_Q(F'\prj{1}(\bar w'_n), q_{n+1}) = 
        \trig_Q(\bar q'_n, q_{n+1}) \in \delta(q_n) = (\delta \circ \rho_Q)(r_n).
    \end{multline*}
    For property (iii) of pre-runs, let $(t_n)_{n \in \omega} \in R^\omega$ be such that $t_0 = r_0$ and $t_{n+1} \in \Base_F(\rho_F(t_n))$ for every $n \in \omega$. If $(t_n)_{n\in\omega} = (r_n)_{n\in\omega}$, we have $(\rho_Q(t_n)) = (q_n)_{n\in\omega} \in \Acc$ by assumption. Otherwise, there exists an $n \in \omega$ such that for all $m \geq n$, we have $t_m \notin \{ r_k \mid k \in \omega \}$. This implies that $(t_m)_{m \geq n}$ is entirely contained in some run $\mathsf R_{z, w}$. By reachability of $\mathsf R_{z,w}$, let $u_0, \dotsc, u_{n-1} \in R_{z,w}$ be such that $u_0$ is the initial state in $R_{z,w}$, $u_{m+1} \in \Base_F(\rho_F(u_m))$ for all $0 \leq m < n-1$, and $t_{n} \in \Base_F(\rho_F(u_{n-1}))$. Since $\mathsf R_{z, w}$ is a run, we have $(\rho_Q(u_m))_{0 \leq m < n} \cdot (\rho_Q(t_m))_{m \geq n} \in \Acc$. Since $\Acc$ is prefix-agnostic, we know $(\rho_Q(t_m))_{m \geq n} \in \Acc$. Again, since $\Acc$ is prefix-agnostic, $(\rho_Q(t_m))_{m \in \omega} \in \Acc$.
\end{proof*}

\begin{proof*}{Proof of Proposition \ref{prop:automatonAlgebraRational}}\qedoverwrite
    Let $\mathcal A = (Q, \delta, Q_I, \Omega)$ be a parity $F$-automaton and $\autalg_{\mathcal A} = (C, \gamma = [\gamma_0, \gamma_1], U)$. Define a two-sorted algebra $(\widehat C, C)$ by:
    \begin{align*}
        \widehat C &\coloneqq \mathcal P(Q \times Q \times \Im(\Omega)), \\
        \widehat c_1 \cdot \widehat c_2 &\coloneqq \{ (q, q_2, max\{m_1, m_2\}) \mid \exists q_1 \in Q: (q, q_1, m_1) \in \widehat c_1 \land (q_1, q_2, m_2) \in \widehat c_2 \}, \\
        \widehat c \times c &\coloneqq \{ q \mid \exists q_1 \in c, m \in \omega: (q, q_1, m) \in \widehat c \, \}, \\
        \Pi((\widehat c_n)_{n\in\omega}) &\coloneqq \{ q_0 \mid \exists (q_n)_{n\in\omega} \in Q^\omega, (m_n)_{n\in\omega} \in \mathbb{N}^\omega: 
        \forall n \in \omega ((q_n, q_{n+1}, m_n) \in \widehat c_n) \: \land 
        \limsup_{n\in\omega} m_n \text{ is even} \}.
    \end{align*}
    By a lengthy but straightforward verification, it can be shown that $(\widehat C, C)$ satisfies the axioms of $\omega$-semigroups. We define the map $\gamma_2: (F'C)^+ \to \widehat C$ by specifying its restriction to the set of generators $F'C$ of the freely generated semigroup $(F'C)^+$. For each $\bar c' \in F'C$, we set:
    \begin{equation*}
        \gamma_2(\bar c') \coloneqq \{ (q, q_1, max\{\Omega(q), \Omega(q_1)\}) \mid \exists \bar q' \in F'Q (\bar q' \inlift \bar c' \land \trig_Q(\bar q', q_1) \in \delta(q) \}.
    \end{equation*}
    We claim that $(\gamma_2, \gamma_1): ((F'C)^+, (F'C)^\omega) \to (\Im(\gamma_2), \Im(\gamma_1))$ is an $\omega$-semigroup homomorphism. It suffices to check preservation of infinite products. For each $(\bar c'_n)_{n\in\omega} \in (F'C)^\omega$ and $q_0 \in Q$, we have:
    \begin{align*}
        q_0 \in \gamma_1((\bar c'_n)_{n\in\omega}) &\iff \\
        \exists (q_n)_{n\geq 1} \in Q^\omega, (\bar q'_n)_{n\in\omega} \in (F'Q)^\omega: \forall n(\bar q'_n \inlift \bar c'_n, \trig_Q(\bar q'_n, q_{n+1}) \in \delta(q_n)),
        \limsup_{n\in\omega} \Omega(q_n) \text{ is even} &\iff \\
        \exists (q_n)_{n\geq 1} \in Q^\omega, (\bar q'_n)_{n\in\omega} \in (F'Q)^\omega: \forall n( \bar q'_n \inlift \bar c'_n, \trig_Q(\bar q'_n, q_{n+1}) \in \delta(q_n)),& \\
        \limsup_{n\in\omega} \max\{\Omega(q_n), \Omega(q_{n+1})\} \text{ is even} &\iff \\
        \exists (q_n)_{n \geq 1} \in Q^\omega, (m_n)_{n \in \omega} \in \omega^\omega: \forall n((q_n, q_{n+1}, m_n) \in \gamma_2(\bar c'_n)),
        \limsup_{n\in\omega} m_n \text{ is even} &\iff \\
        q_0 \in \Pi((\gamma_2(\bar c'_n))_{n\in\omega}). &
        \qedhere
    \end{align*}
\end{proof*}

\section{Detailed Proofs from Section \ref{sec:algebraToAutomaton}}

\begin{proof*}{Proof of Proposition \ref{prop:markingsAndPreruns}}
    Let $\mathcal A_{\mathsf C} = (Q, \delta, Q_I, \Acc)$.
    
    (i) Let $(R,\rho,r_I)$ be a pre-run of $\mathcal A_{\mathsf C}$ on $(X, \xi, x_I)$. We first show that $\mu \coloneqq [\id + \prj{2}] \circ \rho_Q$ satisfies condition (i) of markings, i.e., it is an $F$-coalgebra-to-algebra morphism. Let $r \in R$ and $\bar r = \rho_F(r)$. From property (ii) of pre-runs it follows that $F\rho_Q(\bar r) \in \delta(\rho_Q(r))$. Let $c \coloneqq \mu(r)$. From the definition of the transition function $\delta$ it follows that there exists $\bar c \in C$ such that $\gamma_0(\bar c) = c$ and $F\rho_Q(\bar r) = F\inj{2}(\decomp_C(\bar c))$. Hence
    $F\mu(\bar r) = (F[\id, \prj{2}] \circ F\rho_Q)(\bar r) = (F[\id, \prj{2}] \circ F\inj{2} \circ \decomp_C)(\bar c) = (F\prj{2} \circ \decomp_C)(\bar c)$.
    By Lemma \ref{lem:decompProperties} (i), the latter equals to $\bar c$. Therefore
    $(\gamma_0 \circ F\mu \circ \rho_F)(r) = (\gamma_0 \circ F\mu)(\bar r) = \gamma_0(\bar c) = c = \mu(r)$.

    Next, we prove condition (ii) of markings. Fix arbitrary $(r_n)_{n\in\omega} \in R^\omega$ and $(\bar r'_n)_{n\in\omega} \in (F'R)^\omega$ with $\trig_R(\bar r'_{n+1}, r_{n+1}) = \rho_F(r_n)$ for all $n \in \omega$. Let $\hat R = \{ r \in R \mid \rho_Q(r) \in \inj{2}[F'C \times C] \}$ and $\hat \rho_Q: \hat R \to F'C \times C$ be the map satisfying $\inj{2} \circ \hat \rho_Q = \restr{\rho_Q}{\hat R}$. Consider $(\bar r'_{n+1}, r_{n+1})$ for an arbitrary $n$. Since $F\rho_Q(\rho_F(r_n)) \in \delta(\rho_Q(r_n))$, the definition of $\delta$ implies $F\rho_Q(\rho_F(r_n)) \in (F\inj{2} \circ \decomp_C)[FC]$. Thus $\Base_{F'}(\bar r_{n+1}') \cup \{r_{n+1} \} = \Base_F(\rho_F(r_n)) \subseteq \hat R$. Naturality of $\trig$ gives us:
    $(\trig_{F'C \times C} \circ (F'\hat\rho_Q \times \hat \rho_Q))(\bar r_{n+1}', r_{n+1}) = (F\hat\rho_Q \circ \trig_{\hat R})(\bar r_{n+1}', r_{n+1}) =
    (F\hat\rho_Q \circ \rho_F)(r_n) \in \decomp_C[FC]$.
    By Lemma \ref{lem:decompProperties} (ii), $(F'\prj{2}\circ F'\hat\rho_Q)(\bar r'_{n+1}) = (\prj{1} \circ \hat\rho_Q)(r_{n+1})$.

    By reachability of pre-runs, there exist $s_0, \dotsc, s_m \in R$ such that $s_0 = r_I$, $s_m = r_0$ and $s_{i+1} \in \Base_F(\rho_F(s_i))$. Then $(\rho_Q(s_i))_{i = 0}^{m} \cdot (\rho_Q(r_i))_{i > 0} \in \Acc$ by property (iii) of pre-runs. As a result, $\mu(r_0) = \gamma_1((\prj{1} \circ \hat\rho_Q)(r_n))_{n > 0})$. Now the equality $(F'\prj{2} \circ F'\hat\rho_Q)(\bar r'_{n}) = (\prj{1} \circ \hat\rho_Q)(r_{n})$ implies:
    $\mu(r_0) = \gamma_1(((\prj{1} \circ \hat\rho_Q)(r_n))_{n > 0}) = \gamma_1(((F'\prj{2} \circ F'\hat\rho_Q)(\bar r'_n))_{n > 0}) =
    \gamma_1((F'\mu(\bar r_n'))_{n>0}) = \gamma_1(G\mu((\bar r_n')_{n>0}))$.

    (ii) Let $\mu: X \to C$ be a consistent marking of $(X, \xi)$ with $(C, \gamma)$. Define a pointed $(F \times \Delta_X \times \Delta_Q)$-coalgebra $(R, \langle \rho_F, \rho_X, \rho_Q \rangle, r_I)$ by:
    \begin{align*}
        &R \coloneqq X \times Q, \qquad r_I \coloneqq (x_I, \inj{1} \circ \mu(x_I)), \qquad \rho_X \coloneqq \prj{1}, \qquad \rho_Q \coloneqq \prj{2}, \\
        &\rho_F \coloneqq X \times Q \xrightarrow{\prj{1}} X \xrightarrow{\xi} FX \xrightarrow{\decomp_X} F(F'X \times X) \xrightarrow{F(F'\mu \times \langle \id, \mu \rangle)} F(F'C \times (X \times C)) \\
        &\hspace{167px} \xrightarrow{\cong} F(X \times (F'C \times C)) \xrightarrow{F(\id \times \inj{2})} F(X \times Q).
    \end{align*}
    Let $\dot{\mathsf R} \coloneqq (\dot R, \langle \rho_F, \rho_X, \rho_Q \rangle, r_I)$ be its reachable subcoalgebra\footnote{This is a slight abuse of notation -- we write $\rho_F$ but we mean the restriction of $\rho_F$ to $\dot R$. Similarly for $\rho_X$ and $\rho_Q$.}. We prove that $\dot{\mathsf R}$ is a pre-run.

    For property (i) of pre-runs, we are to show that $F\rho_X \circ \rho_F = \xi \circ \rho_X$. For an arbitrary $(x, q) \in \dot R$ we have:
        $(F\rho_X \circ \rho_F)(x, q) = (F\prj{1} \circ \rho_F)(x,q) = (F\prj{2} \circ \decomp_X \circ \xi)(x) = (\id \circ \xi)(x) =
        (\xi \circ \rho_X)(x, q)$,
    where the second equality uses the definition of $\rho_F$ and the third equality uses Lemma \ref{lem:decompProperties} (i).

    For property (ii) of pre-runs, we first prove that every $(x, q) \in \dot R$ satisfies $[\id, \prj{2}](q) = \mu(x)$. We do this by induction on the successor relation of $\dot{\mathsf R}$.
    \begin{itemize}
        \item Root case: by definition of $r_I = (x_I, \inj{1} \circ \mu(x_I))$.
        \item Successor case: let $(x, q) \in \dot R$, then the equality holds for an arbitrary element of $\Base_{F \times \Delta_X \times \Delta_Q}(\langle \rho_F, \rho_X, \rho_Q \rangle(x, q)) = \Base_F(\rho_F(x, q))$, because:
        \begin{equation}
            \Base_F(\rho_F(x, q)) =
            \{ (x_1, \inj{2}(F'\mu(\bar x'), \mu(x_1)) \mid (\bar x', x_1) \in \Base_F((\decomp_X \circ \xi)(x)) \}. \label{eq:SucRhoF}
        \end{equation}
    \end{itemize}
    Now we prove that if $(x, q) \in \dot R$, then $F\rho_Q(\rho_F(x, q)) \in \delta(\rho_Q(x, q))$. We have
    $(F\rho_Q \circ \rho_F)(x,q) = (F\prj{2} \circ \rho_F)(x, q) = (F(F'\mu \times \mu) \circ \decomp_X \circ \xi)(x) = (\decomp_C \circ F\mu \circ \xi)(x)$.
    Hence there exists $\bar c \coloneqq (F\mu \circ \xi)(x)$ such that
    $\gamma_0(\bar c) = (\gamma_0 \circ F\mu \circ \xi)(x) = \mu(x) = [\id,\prj{2}](q)$ and $(F\rho_Q \circ \rho_F)(x, q) = \decomp_C(\bar c)$.
    According to the definition of $\delta$, this implies $(F\rho_Q \circ \rho_F)(x, q) \in \delta(\rho_Q(x, q))$.

    For property (iii) of pre-runs, let $(r_n)_{n\in \omega} \in {\dot R}^\omega$, with $r_n = (x_n,q_n)$, be such that $r_0 = r_I$ and $r_{n+1} \in \Base_F(\rho_F(r_n))$ for all $n \in \omega$. We are to show that $(\rho_Q(r_n))_{n\in\omega} \in \Acc$, i.e., $(q_n)_{n\in\omega} \in \Acc$. By Equation \eqref{eq:SucRhoF}, for all $n \in \omega$ there exists $\bar x'_{n+1} \in F'X$ such that $(\bar x'_{n+1}, x_{n+1}) \in \Base_F(\decomp_X(\xi(x_n))$ and $q_{n+1} = \inj{2}(F'\mu(\bar x'_{n+1}), \mu(x_{n+1}))$. By Lemma \ref{lem:decompProperties} (iii), we have $\trig_X(\bar x'_{n+1}, x_{n+1}) = \xi(x_n)$. Hence, for every $m \in \omega$, property (ii) of the marking $\mu$ implies that $\gamma_1(G\mu((\bar x'_{n+1})_{n > m})) = \mu(x_m)$. Finally, by letting $c_0 \coloneqq \mu(x_0)$ and $(\bar c'_{n+1}, c_{n+1}) \coloneqq (F'\mu(\bar x'_{n+1}), \mu(x_{n+1}))$, we have that $(q_n)_{n\in \omega}$ can be written as $\inj{1}(c_0) \cdot (\inj{2}(\bar c'_n,c_n))_{n>0}$ with $\forall m(c_m = \gamma_1((\bar c_n')_{n>m})$. Therefore $(q_n)_{n\in\omega} \in \Acc$.
\end{proof*}

\begin{proof*}{Proof of Lemma \ref{lem:markingsPreservedByPrecomp}}
    Condition (i) of markings for $\mu \circ f$ follows immediately from the commutation of the rectangle:
    \[\begin{tikzcd}
        Y & X & C \\
        FY & FX & FC
        \arrow["f", from=1-1, to=1-2]
        \arrow["\upsilon"', from=1-1, to=2-1]
        \arrow["\mu", from=1-2, to=1-3]
        \arrow["\xi", from=1-2, to=2-2]
        \arrow["Ff"', from=2-1, to=2-2]
        \arrow["{F\mu}"', from=2-2, to=2-3]
        \arrow["{\gamma_0}"', from=2-3, to=1-3]
    \end{tikzcd}\]
    where the two squares commute by assumption. To show that $\mu \circ f$ satisfies condition (ii), let $(y_n) \in Y^\omega$ and $(\bar y_n')_{n \in \omega} \in GY$ satisfy $\trig_Y(\bar y'_{n+1}, y_{n+1}) = \upsilon(y_n)$ for all $n \in \omega$. Let $x_n \coloneqq f(y_n)$ and $\bar x_n' \coloneqq F'f(\bar y'_n)$ for every $n \in \omega$ and $n > 0$, respectively. We have
    $\trig_X(\bar x_{n+1}', x_{n+1}) = (Ff \circ \trig_Y)(\bar y'_{n+1}, y_{n+1}) = (Ff \circ \upsilon)(y_n) = (\xi \circ f)(y_n) = \xi(x_n)$.
    Since $\mu$ satisfies condition (ii) of markings, we get $\mu(x_0) = (\gamma_1 \circ Gf)((\bar x'_n)_{n > 0})$. Therefore
    $(\mu \circ f)(y_0) = \mu(x_0) = (\gamma_1 \circ Gf)((\bar x'_n)_{n > 0}) = (\gamma_1 \circ G(\mu \circ f))((\bar y'_n)_{n > 0})$.
\end{proof*}

The following lemma collects properties of the $(F+G)$-coalgebra $(\thin Z, \eta)$ (see Definition \ref{def:etaMap}) that will be used towards proving uniqueness of markings.

\begin{lemma}
\label{lem:etaProperties}
    \begin{enumerate}
        \item[(i)] $\iota: (\thin Z, \eta) \to (A, \alpha^{-1})$ is an $(F+G)$-coalgebra morphism;
        \item[(ii)] $\thin \beta \circ \eta = \id$;
        \item[(iii)] if $\eta(z) \in \inj1[F\thin Z]$, then $\eta(z) = \inj1(\thin \zeta(z))$, for all $z \in \thin Z$;
        \item[(iv)] if $\eta(z) = \inj2((z_n)_{n \in \omega}) \in \inj2[G\thin Z]$, then $\thin \beta_1((z_n)_{n > m}) \notin \Base_{F'}(\bar z'_m)$ for all $m \in \omega$, for all $z \in \thin Z$.
    \end{enumerate} 
\end{lemma}
\begin{proof}
    Properties (i), (ii) and (iii) are straightforward. For (iv), we use the fact that $\beta_1((z_n)_{n > m})$ has the same \emph{major rank} as $\beta_1((z_n)_{n \geq m})$~\cite[Lemma~A.2]{ChernevCirsteaHansenKupkeThinCoalg}, while any element of $\Base_{F'}(\bar z'_m)$ has strictly lower major rank~\cite[Observation~V.5]{ChernevCirsteaHansenKupkeThinCoalg}.
\end{proof}

\begin{proof*}{Proof of Lemma \ref{lem:F+GcoalgebraOnX}}
    We first define the map $\upsilon$ and then prove properties (i) and (ii). Let $x \in X$ and $z \coloneqq \tbeh(x)$. We know that either $\eta(z) \in \inj{1}[F\thin Z]$ or $\eta(z) \in \inj{2}[G\thin Z]$. In the former case, define $\upsilon(x) = \inj1(\xi(x))$. In the latter case, we have $z = \thin \beta_1((\bar z'_n)_{n > 0})$ for $\eta(z) = \inj2(\bar z'_n)_{n > 0}$. We let $z_n \coloneqq \thin \beta_1((\bar z'_m)_{m > n})$ and observe the equality $\thin \zeta(z_n) = \trig_{\thin Z}(\bar z'_{n+1}, z_{n+1})$, by coherence of $(\thin Z, \thin \beta)$ and $\thin \beta_0 = (\thin \zeta)^{-1}$. Our goal is to define $(x_n)_{n \in \omega}$ and $(\bar x'_n)_{n > 0}$ such that $x_0 = x$, $\tbeh(x_n) = z_n$, $F'\tbeh(\bar x'_{n+1}) = \bar z'_{n+1}$ and $\trig_X(\bar x'_{n+1}, x_{n+1}) = \xi(x_n)$, for all $n \in \omega$. We proceed by induction on $n \in \omega$. 
    
    For the base case, $x_0 \coloneqq x$. For the inductive step, suppose $x_n$ has been defined and $\tbeh(x_n) = z_n$, we define $x_{n+1}$ and $\bar x'_{n+1}$. Since $\tbeh: (X,\xi) \to (\thin Z, \thin \zeta)$ is an $F$-coalgebra morphism, we have $\mathcal P \tbeh(\Base_F(\xi(x_n))) = \Base_F(\thin \zeta(z_n))$. Since $z_{n+1} \in \Base_F(\thin \zeta(z_n))$, there exists $x \in \Base_F(\xi(x_n))$ with $\tbeh(x) = z_{n+1}$. We take $x_{n+1}$ to be any such $x$. By the properties of $\trig$, there exist $\bar x'_{n+1} \in F'X$ with $\trig_X(\bar x'_{n+1}, x_{n+1}) = \xi(x_n)$. It remains to show $F'\tbeh(\bar x'_{n+1}) = \bar z'_{n+1}$. We have $\trig_{\thin Z}(F'\tbeh(\bar x'_{n+1}), z_{n+1}) = (F\tbeh \circ \trig_X)(\bar x'_{n+1}, x_{n+1}) = (F\tbeh \circ \xi)(x_n) =(\thin \zeta \circ \tbeh)(x_n) = \thin \zeta(z_n) = \trig_{\thin Z}(\bar z'_{n+1}, z_{n+1})$.
        Moreover, Lemma \ref{lem:etaProperties} (iv) tells us that $z_{n+1} \notin \Base_{F'}(\bar z'_{n+1})$. By the properties of $\trig$, we infer $F'\tbeh(\bar x'_{n+1}) = \bar z'_{n+1}$.

    Next, we verify property (i) of the Lemma. Let $x \in X$, we show $((F+G)\tbeh \circ \upsilon)(x) = (\eta \circ \tbeh)(x)$. If $\upsilon(x) \in \inj1[FX]$, then 
    $((F+G)\tbeh \circ \upsilon)(x) = ((F+G)\tbeh \circ \inj1 \circ \xi)(x) = (\inj1 \circ F\tbeh \circ \xi)(x) =
        (\eta \circ \tbeh)(x)$.
    If $\upsilon(x) \in \inj2[F(X)]$, then
    $((F+G)\tbeh \circ \upsilon)(x) = (\inj2 \circ G\tbeh)((\bar x'_n)_{n>0}) = \inj2((\bar z'_n)_{n > 0}) = (\eta \circ \tbeh)(x)$.

    Finally, we verify property (ii) of the Lemma. Let $\mu: X \to C$ be a marking of $(X, \xi)$ with $(C, \gamma)$ and $x \in X$, we show that $\gamma \circ (F+G)\mu \circ \upsilon = \mu$. If $\upsilon(x) \in \inj1[FX]$, then $(\gamma \circ (F+G)\mu \circ \upsilon)(x) = (\gamma \circ (F+G)\mu)(\inj{1}(\xi(x))) = \mu(x)$,
    because $\mu: (X, \xi) \to (C, \gamma_0)$ is an $F$-coalgebra-to-algebra morphism. If $\upsilon(x) \in \inj2[GX]$, by the construction of $\upsilon$, there exist $(x_n)_{n\in\omega} \in X^\omega$ and $(\bar x'_n)_{n>0} \in GX$ such that $\upsilon(x) = \inj2((\bar x'_n)_{n>0})$ and $\trig_X(\bar x'_{n+1}, x_{n+1}) = \xi(x_n)$. Now
    $(\gamma \circ (F+G)\mu \circ \upsilon)(x) = (\gamma \circ (F+G)\mu)(\inj{2}((\bar x'_n)_{n>0})) =
        (\gamma_1 \circ G\mu)((\bar x'_n)_{n>0}) = \mu(x_0) = \mu(x)$,
    where the third equality uses property (ii) of the marking $\mu$.
\end{proof*}

\begin{proof*}{Proof of Lemma \ref{lem:preRunUniquelyDetermined}}
    By Proposition \ref{prop:markingsAndPreruns} (i), $[\id,\prj{2}] \circ \rho_Q: R \to C$ is a marking of $(R, \rho_F)$ with $(C, \gamma)$. By Propositions \ref{prop:existenceMarkings} and \ref{prop:uniquenessOfMarkings}, $[\id,\prj{2}] \circ \rho_Q = \cev_{(C, \gamma)} \circ \tbeh_{(X, \xi)}$. Similarly, $[\id,\prj{2}] \circ \rho_Q' = \cev_{(C, \gamma)} \circ \tbeh_{(X, \xi)}$, so $[\id,\prj{2}] \circ \rho_Q = [\id,\prj{2}] \circ \rho_Q'$. We prove by induction on the successor relation of $(R, \rho_F, r_I)$ that $\rho_Q(r) = \rho_Q'(r)$, for all $r \in R$.
    \begin{itemize}
        \item $r = r_I$. Then $\rho_Q(r) = \inj{1}(c)$ and $\rho_Q'(r) = \inj{1}(c')$ for some $c, c' \in C$. We have $c = ([\id,\prj{2}] \circ \rho_Q)(r) = ([\id,\prj{2}] \circ \rho_Q')(r) = c'$, so $\rho_Q(r) = \rho_Q'(r)$.
        \item $r \in \Base_F(\rho_F(s))$. Define $Y \coloneqq \Base_F(\rho_F(s))$ and $\bar r \coloneqq \rho_F(s) \in FY$. Let $\delta$ be the transition function of $\mathcal A_{\mathsf{C}}$. Then $F\rho_Q (\bar r) \in (\delta \circ \rho_Q)(s)$, so $F\rho_Q (\bar r) = (F\inj{2} \circ \decomp_C)(\bar c)$ for some $\bar c \in FC$. Hence there exists $\hat \rho_Q: Y \to F'C \times C$ such that $\rho_Q(r_1) = (\inj{2} \circ \hat \rho_Q)(r_1)$ for every $r_1 \in Y$.

        Let $\bar r' \in F'Y$ be such that $\trig_Y(\bar r', r) = \bar r$. We have $\trig_{F'C \times C}(F'\hat \rho_Q(\bar r'), \hat \rho_Q(r)) = (F\hat \rho_Q \circ \trig_Y)(\bar r', r) \in \decomp_{C}[FC]$,
        so Lemma \ref{lem:decompProperties} (ii) implies $F'\prj{2}(F'\hat\rho_Q(\bar r')) = \prj{1}(\hat \rho_Q(r))$. We get
        $\rho_Q(r) = \inj{2}((\prj{1} \circ \hat\rho_Q)(r), (\prj{2} \circ \hat\rho_Q)(r)) = 
        \inj{2}(F'(\prj{2} \circ \hat\rho_Q)(\bar r'), \prj{2} \circ \hat \rho_Q(r)) =
        \inj{2}(F'(\cev_{(C,\gamma)} \circ \tbeh_{(X,\xi)})(\bar r'), (\cev_{(C,\gamma)} \circ \tbeh_{(X,\xi)})(r))$.
        Analogously, $\rho_Q'(r) = \inj{2}(F'(\cev_{(C,\gamma)} \circ \tbeh_{(X,\xi)})(\bar r'), (\cev_{(C,\gamma)} \circ \tbeh_{(X,\xi)})(r))$, so $\rho_Q(r) = \rho_Q'(r)$. \qed
    \end{itemize}
    \qedoverwrite
\end{proof*}

\section{Detailed Proofs from Section \ref{sec:mainResults}}

\begin{proof*}{Proof of Lemma \ref{lem:algebraicAutomatonOfRationalIsOmegaRegular}}
    Let $\gamma = [\gamma_0, \gamma_1]$ and $(\gamma_2, \gamma_1): ((F'C)^+, (F'C)^\omega) \to (\widetilde C, C)$ be an $\omega$-semigroup quotient witnessing rationality of $(C, \gamma)$. Let $Q = C + (F'C \times C)$ and $\Acc$ denote the states and the acceptance condition of $\mathcal A_{\mathsf{C}}$, respectively. For each $c \in C$, we know that the language
        $L(c) \coloneqq \{ (\bar c'_n)_{n\in\omega} \in (F'C)^\omega \mid \gamma_1((\bar c'_n)_{n\in\omega}) = c \}$
    is $\omega$-regular. Indeed, $L(c)$ is recognised by the finite $\omega$-semigroup $(\widetilde C, C)$. Thus let $\phi_c$ be a monadic second-order formula defining $L(c)$. We can obtain a formula $\phi_c(n)$ defining 
    $L_n(c) \coloneqq Q^n \cdot (\inj{2})^\omega(L(c) \times C^\omega)$
    by adjusting the alphabet from $F'C$ to $Q$, by replacing all occurrences of the constant $0$ in $\phi_c$ with $n + 1$ and by restricting all first-order quantifiers to $> n$. Now consider the formula:
    \begin{equation*}
        \phi \coloneqq \forall n \bigwedge_{c\in C} \bigwedge_{\bar c' \in F'C} (n \in \inj{2}(\bar c', c) \to \phi_c(n)).
    \end{equation*}
    which defines the set of infinite words $(q_n)_{n \in \omega} \in Q^\omega$ with the property that if $q_n$ is of the form $\inj{2}(\bar c', c)$ then $(q_n)_{n \in \omega} \in L_n(c)$.
    By intersecting the language of $\phi$ with the $\omega$-regular language $\inj{1}[U] \cdot (\inj2)^\omega(F'C \times C)$,
    thus ensuring that that the first letter of any word in the language lies in the recognising set $U$,
    we obtain exactly $\Acc$. Therefore $\Acc$ is $\omega$-regular.
\end{proof*}

\begin{proof*}{Proof of Lemma \ref{lem:algebraicAutomatonPrefixAgnostic}}
    Let $\algaut_{\mathsf C} = \mathcal A \coloneqq (Q, \delta, Q_I, \Acc)$. Define $\mathcal A' = (Q, \delta, Q_I, \Acc')$ with:
    \begin{multline*}
        \Acc' \coloneqq \{ (q_n)_{n \in \omega} \in Q^\omega \mid \exists m \in \omega, (\bar c'_n)_{n \geq m} \in F'C, (c_n)_{n \geq m} \in C^\omega: \\
        (q_n)_{n \geq m} = (\inj{2}(\bar c_n', c_n))_{n \geq m}  \land
        \forall k \geq m (c_m = \gamma_1((\bar c'_n)_{n \geq k})) \}.
    \end{multline*}
    It follows from the definition of $\Acc'$ that it is prefix-agnostic. Since $\Acc \subseteq \Acc'$, we have that every run of $\mathcal A$ is a run of $\mathcal A'$. Conversely, let $\mathsf R = (R, \rho = \langle \rho_F, \rho_X, \rho_Q \rangle, r_I)$ be a run of $\mathcal A'$ on some pointed coalgebra. In order to prove that $\mathsf R$ is a run of $\mathcal A$, it suffices to demonstrate that property (iii) of runs holds, as the other properties hold automatically. Let $(r_n)_{n \in \omega} \in R^\omega$ satisfy $r_{n+1} \in \Base_F(\rho_F(r_n))$ for all $n \in \omega$. From property (ii) of runs we have $(F\rho_Q \circ \rho_F)(r_n) \in (\delta \circ \rho_Q)(r_n)$, so $(F\rho_Q \circ \rho_F)(r_{n+1}) = (F\inj{2} \circ \decomp_C)(\bar c_n)$ for some $\bar c_n \in FC$. In addition, $\rho_Q(r_I) \in Q_I = \inj{1}[U]$, so we deduce that $(\rho_Q(r_n))_{n \in \omega} = \inj{1}(c_0) \cdot (\inj{2}(\bar c'_n, c_n))_{n \geq 1}$ for some $(\bar c'_n)_{n \geq 1} \in GC$ and $(c_n)_{n \geq 1} \in C^\omega$, and $\delta(\bar c_n) = c_n$. Furthermore, $(\bar c'_{n+1}, c_{n+1}) \in \Base_{F}(\decomp_C(\bar c_n))$, so by Lemma \ref{lem:decompProperties} (iii), $\trig(\bar c'_{n+1}, c_{n+1}) = \bar c_n$. Now, to show that $(\rho_Q(r_n))_{n \in \omega} \in \Acc$, we take $m \in \omega$ and show $c_m = \gamma_1((\bar c'_n)_{n > m})$. By assumption, $(\rho_Q(r_n))_{n \in \omega} \in \Acc'$, so there exists $k \in \omega$ such that for all $l \geq k$, we have $c_l = \gamma_1((\bar c'_n)_{n > l})$. If $m \geq k$, we are done. Otherwise, we proceed by induction on $m = k - 1, k-2, \dotsc, 0$. Suppose we have shown $c_{m+1} = \gamma_1((\bar c'_n)_{n > m + 1})$. By coherence of $(C, \gamma)$, $\gamma_1((\bar c'_n)_{n > m}) = \gamma_0(\bar c'_{m + 1}, \gamma_1((\bar c'_n)_{n > m + 1})) = \gamma_0(\bar c'_{m+1}, c_{m+1}) = c_m$.
    Thus $(\rho_Q(r_n))_{n \in \omega} \in \Acc$.
\end{proof*}
 
\end{document}